\documentclass[12pt,reqno]{amsart}

\usepackage{amsmath}
\usepackage{amssymb}
\usepackage{mathabx}
\usepackage{mathrsfs}               % \mathscr
\usepackage{dsfont}		% \1 
\usepackage{fancyhdr}
\usepackage{a4wide}		%margin
\usepackage[colorlinks=true,linkcolor=blue]{hyperref}%
\usepackage{mathtools}% http://ctan.org/pkg/mathtools
\usepackage[square,sort,compress,comma,numbers]{natbib}

\usepackage{caption}
\usepackage{ulem}
\usepackage[colorinlistoftodos,prependcaption,textsize=tiny]{todonotes}

\usepackage{pgfplots}
\usetikzlibrary{patterns}

\usepackage{mathtools}\usepackage{mathtools}			% underbrace 

\usepackage{bm}		%bold 
\usepackage{braket}			% bra-ket 

\numberwithin{equation}{section}	 % Equations are numbered according to Section

\usepackage{enumitem}						% Enumerate package

\newcommand{\Q}{ {\mathcal Q}  }

\newcommand{\calN}{\mathcal N}

\renewcommand{\d}{\mathrm{d}}							% Roman letters
\newcommand{\re}{\mathrm{Re}}

\newcommand{\T}{\mathbb T}
\newcommand{\To}{    {\mathbf T}  }
\newcommand{\N}{ \mathbf N} 		% Natural numbers
\newcommand{\C}{\mathbb{C}}			 % Complex numbers
\newcommand{\Z}{ \mathbf Z} 		% Integer numbers
\newcommand{\R}{  {     \mathbf  R } } 			% Real numbers

\renewcommand{\S}{\mathcal S}

\newcommand{\V}{\mathbb V}
\newcommand{\F}{\mathscr{F}}

\renewcommand{\H}{\mathbb H }

\newcommand{\ho}{h_0}

\newcommand{\1}{\mathds{1}}
\newcommand{\<}{\left\langle}							% Miscellanea
\renewcommand{\>}{\right\rangle}
\renewcommand{\le}{\leqslant}
\renewcommand{\leq}{\leqslant}
\renewcommand{\ge}{\geqslant}
\renewcommand{\geq}{\geqslant}     
\newcommand{\vp}{\varphi}
\newcommand{\ve}{\varepsilon}

\newcommand{\supp}{\textnormal{supp}}
\renewcommand{\t}[1]{\textnormal{#1}}
\renewcommand{\(}{\left(}
\renewcommand{\)}{\right)}
\renewcommand{\[}{   [  }

\newcommand{\W}{W}
\newcommand{\D }{\mathcal D}
\newcommand{\K}{\mathscr K}
\newcommand{\RR}{\mathbb R}

\newcommand{\eff}{\t{eff}}
\makeatletter
\renewcommand\subsection{\@startsection{subsection}{2}%
	\z@{-0.5\linespacing\@plus-.5 \linespacing}{.3\linespacing}%
	{\bfseries \scshape }}
\renewcommand\subsubsection{\@startsection{subsubsection}{2}%
	\z@{-0.3\linespacing\@plus-.5 \linespacing}{.3\linespacing}%
	{\itshape   }}
\newtheorem{condition}{Condition}

\newtheorem{theorem}{Theorem} 
\newtheorem{proposition}{Proposition}[section]
\newtheorem{corollary}{Corollary}[section]
\newtheorem{lemma}{Lemma}[section]
\theoremstyle{remark}
\theoremstyle{definition}

\newtheorem{remark}{Remark}[section]
\newtheorem{remarks}{Remarks}[section]

\newcommand{\kf}{k_{\rm F}}

\begin{document}
	
	%		\date{\today} 
	% \title[On the stability-instability transition in large Bose-Fermi mixtures]{On the stability-instability transition in large Bose-Fermi mixtures}
	\title[Emergence of fermion-mediated interactions in Bose-Fermi mixtures]{Emergence of fermion-mediated interactions in Bose-Fermi mixtures}

	\author{Esteban C\'ardenas}
	\address[Esteban C\'ardenas]{Department of Mathematics,
		University of Texas at Austin,
		2515 Speedway,
		Austin TX, 78712, USA}
	\email{eacardenas@utexas.edu}
	
	\author{Joseph K. Miller}
	\address[Joseph K. Miller]{Department of Mathematics, Stanford University, 450 Jane Stanford Way, Stanford, 94305, USA}
	\email{: jkm314@stanford.edu}
	
	\author{David Mitrouskas}
	\address[David Mitrouskas]
	{Institute of Science and Technology Austria (ISTA), Am Campus 1, 3400 Klosterneuburg, Austria}
	\email{mitrouskas@ist.ac.at}
	
	\author{Nata\v sa Pavlovi\' c}
	\address[Nata\v sa Pavlovi\'c]{Department of Mathematics,
		University of Texas at Austin,
		2515 Speedway,
		Austin TX, 78712, USA}
	\email{natasa@math.utexas.edu} 
	
	\frenchspacing
	
	\begin{abstract}
	This work is inspired by recent experimental observations in ultracold atomic Bose-Fermi mixtures [DeSalvo et al., Nature 568 (2019)]. These experiments reveal the emergence of an attractive fermion-mediated interaction between bosons, as well as a stability-instability transition. We give the first mathematical demonstration of this transition by studying the low-energy spectrum of a many-body interspecies Hamiltonian. More precisely, we show the convergence of its eigenvalues towards those of an effective Bose Hamiltonian, which includes fermion-mediated effects. Applying this result to a model with short-range potentials, we derive a stability-instability transition in the bosonic subsystem, driven by the Bose–Fermi coupling strength $g$. For small $|g|$, the bosons form a stable Bose–Einstein condensate with the energy per particle uniformly bounded from below. For large $|g|$, the energy per particle is no longer uniformly bounded from below, signaling the collapse of the condensate.
		% We give the first mathematical demonstration of both phenomena by studying  the low-energy spectrum of a many-body Hamiltonian. 
		% More precisely, in the chosen scaling, the fermions induce an effective attraction among the bosons, which competes with their intrinsic repulsive interaction. Our main result demonstrates the convergence of the eigenvalues towards those of an effective Bose Hamiltonian. Applying this result to a model with short-range potentials, we derive a stability-instability transition in the bosonic subsystem, driven by the Bose–Fermi coupling strength $g$. For small $|g|$, the bosons form a stable Bose–Einstein condensate with the energy per particle uniformly bounded from below. For large $|g|$, the energy per particle is no longer uniformly bounded from below, signaling the collapse of the condensate.
	\end{abstract}
	
	\maketitle

	{\hypersetup{linkcolor=black}
		\tableofcontents}
	
	\section{Introduction}
	
	In both high-energy and condensed matter physics, particle exchange processes play a central role in understanding the emergence of effective forces. 
	Notable examples include the Yukawa potential  arising from the exchange of massive bosons, and phonon-mediated Cooper pairing in superconductors. Ultracold atomic systems, particularly Bose-Fermi mixtures, offer a versatile and controllable platform for investigating analogous phenomena in novel regimes. With recent advancements in atomic physics---such as sympathetic laser cooling and magneto-optical trapping---it has become possible to simultaneously cool down and trap bosonic and fermionic gases \cite{Schreck01a,Schreck01b}, paving the way for a better understanding of the interplay between mediated interactions and collective behaviour in interesting physical systems.
	
	The experimental realization of Bose-Fermi mixtures has revealed a remarkably rich phase diagram, involving polaron formation \cite{Massignan14,Scazza22} mediated interactions \cite{Kinnunen15,DeSalvo19,Huang20,
		Zheng21,Shen24}, novel pairing mechanisms \cite{Wu16,Kinnunen18} and states of dual superfluidity \cite{Yao16,Roy17}. Our interest in this topic was motivated by recent experiments by DeSalvo et al. \cite{DeSalvo19} on mixtures of light fermionic $\textnormal{Li}^6$ atoms and heavy bosonic $\textnormal{Cs}^{133}$ atoms. In this mixture, the Cs atoms form a Bose-Einstein condensate (BEC) which is completely immersed in a much larger degenerate Fermi gas of Li atoms. A key observation---anticipated also by theoretical predictions \cite{Tsurumi2000, Chui2004, Santamore08}---is the emergence of an attractive fermion-mediated interaction between the bosons. As confirmed by the experimental data in \cite{DeSalvo19}, this effective interaction creates a dichotomy in the behavior of the bosonic subsystem. For weak interspecies coupling, the bosons form a stable BEC with a shifted energy, coexisting with the degenerate Fermi gas. However, as the coupling strength increases, the mediated attraction surpasses the intrinsic boson-boson repulsion, eventually destabilizing the BEC and leading to a collapse of the condensate cloud.
	
	The aim of this paper is to give a rigorous derivation of  fermion-mediated interactions--starting from an interspecies many-body Hamiltonian--and to prove the associated stability-instability transition.

	\subsection{The model}
	We consider
	$N$ bosons and $M$ fermions
	moving in   the three-dimensional torus   
	$    \To ^3     =     (\R / 2\pi \Z )^3 $.
	The Hilbert space of the system is the tensor product 
	\begin{equation}
		\mathscr H = \mathscr H_{ \rm B } \otimes \mathscr H_{\rm F} 
		\qquad
		\t{where}
		\qquad 
		\mathscr H _{ \rm B}  =  \bigotimes_{{\rm sym}}^N  L^2 (     \To^3 , \d x  ) 
		\quad 
		\t{and}
		\quad
		\mathscr H _{\rm  F} = \bigwedge^M L^2 (     \To^3 , \d y   ) \ , 
	\end{equation}
	which account for bosonic and fermionic statistics, respectively. 
	For simplicity, we assume both species are spinless, and of equal mass. 
	The  microscopic 	  Hamiltonian 
	in appropriate  units then takes the form 
	\begin{align}
		\label{eq:H}
		H
		\,  = 	\, 
		h  \otimes \1  
		\,  		 + 	\, 
		\1 \otimes  		  \sum_{	1 \leq j \leq  M 	} 
		(-\Delta_{y_j})   
		\, + \, 
		\lambda \sum_{ 1 \leq i \leq N } \sum_{1 \leq j \leq M}
		V(x_i-y_j) 
	\end{align}
	with Bose Hamiltonian
	\begin{align} \label{eq:HB}
		h  \ =\ 	\sum_{ 1 \leq i \leq N }
		(-\Delta_{x_i}) \ +  \  \frac{1}{N} \sum_{1 \le i<j \le N}  
		\W 
		(x_i-x_j)   \ . 
	\end{align}
	Here, $\W$ and $V$ are sufficiently regular, even functions describing Bose-Bose and Bose-Fermi interactions, respectively. 
	To capture different physical situations---particularly short-range interactions in dilute atomic gases---we allow the potentials to depend on the number of bosons $N$.
	We also adopt the common and well-justified approximation of neglecting short-range interactions between fermions \cite{GriesemerH21}.
	The coupling $\lambda$ is specified in \eqref{scaling} below. 
	
	We are interested in describing the low-energy spectrum of the operator $H$ in a regime  where the fermionic energy is much larger than the energy of the bosons. 
	In order to make this limit precise,  
	we introduce  on the momentum lattice $ (     \To^3)^* = \Z^3$ the
	\textit{Fermi ball }
	\begin{equation}
		B_{\rm F } = \{ 	   k\in \Z^3 : | k | \leq k_{\rm F } 		\}
	\end{equation}
	where $k_{\rm F }  \geq 1  $ is the \textit{Fermi momentum}.
	We follow the convention that  $\kf  $ is  a given  number, and we 
	set the fermionic particle number as
	\begin{equation}
		M = |B_{\rm F }|  \ . 
	\end{equation}
	In particular, 
	thanks to simple geometric asymptotics, 
	we know that $M \sim \frac{4\pi}{3 }   \kf^3$
	for large $ \kf $. 
	We further introduce the Fermi energy
	\begin{equation}
		E_{\rm F}  =   \sum_{   k \in B_{\rm F }	} k^2    \ 
	\end{equation}
	that satisfies $E_{\rm F } \sim \frac{4\pi}{5 }  \kf^5$
	for large $ \kf . $
	The Fermi energy is the ground state energy of the 
	non-interacting fermionic subsystem, 
	and  the ground state
	is given by the Slater determinant 
	of plane waves 
	\begin{equation}
		\Omega_{\rm F }  = \textstyle  \bigwedge_{ k \in B_{\rm F }}
		e_k 
		\qquad 
		\t{where}
		\qquad e_k(x) = (2\pi)^{	- \frac{3}{2}}e^{ik\cdot x} \ , \qquad k \in \Z^3  \ .
		\label{plane-wave}
	\end{equation}
	We will be interested in the regime for which
	the fermionic energy $E_{\rm F }$ 
	is much larger than the bosonic energy.
	
	The coupling  $\lambda $  of the Bose-Fermi potential in \eqref{eq:H} is fixed throughout the article as
	\begin{align}\label{scaling}
		\lambda\ = \ \frac{1}{\sqrt{ 4\pi N \kf  } }\ .
	\end{align}
	This choice leads to an effective  Bose-Bose interaction       which competes   with  $W$. 
	Let us   explain the heuristics (as described, for instance, in \cite{DeSalvo19}).  
	For low-energy states,  we expect most  fermions to occupy   the Fermi ball $B_{\rm F }$. 
	Moreover, as $ \kf  \to \infty$, there is a separation of energy scales between fermions and bosons, which causes an effective (adiabatic) decoupling of the two subsystems. However, bosons can still interact with fermions through coherent three-particle processes: one boson excites a fermion with momentum change $k$, while another boson interacts with the same fermion with momentum change $-k$. These processes do not establish fermion-boson correlations (see e.g. Corollary \ref{coro} below), but they induce an effective interaction between boson pairs (see Figure \ref{fig1} for a visualization). The amplitude of this process scales with the number of fermions a boson can interact with. Due to the Pauli exclusion principle,  only fermions near the   surface  of the Fermi ball can participate.  Thus,  the number of contributing fermions scales like $  \kf^2$.  Energy conservation (which appears in the form of a free fermionic resolvent) reduces this number further to $\kf$. Finally, we note that as a second order process,  the amplitude is proportional to $\lambda^2$. Consequently, the strength of the effective potential between two bosons scales like $\lambda^2 \kf$. To ensure this interaction is comparable to $\frac{1}{N} W$ in \eqref{eq:HB}, we set $\lambda$ as in \eqref{scaling}. The factor $1/\sqrt{4\pi}$ is    introduced only  for   convenience.  
	
	Despite  significant interest in the physics community, the mathematical analysis of Bose-Fermi mixtures remains largely unexplored. While purely bosonic (e.g. \cite{Yin2010,Seiringer2011,NRS2016,LSY2000,
		Lieb:book,LNR14,Fournais2023,Fournais2020,
		Brooks,BLPR2023,BBCS2019,BBCS2020,
		BBCS2020b}) and purely fermionic (e.g. \cite{ HPR,Christiansen2024,Christiansen2023-1,
		Christiansen2023-2,Benedikter3,Benedikter2,
		Benedikter1}) many-body systems have been extensively studied, the literature on combined mixtures is scarce. Before presenting our results, we briefly mention two works on the dynamics of Bose-Fermi mixtures: In \cite{Cardenas2023}, 
	the first, second and last author of this paper 
	focus on a semi-classical regime where bosons and fermions remain effectively coupled through mean-field potentials. Starting from an interspecies many-body Hamiltonian with an appropriate mean-field scaling, 
	a system of coupled  Hartree-Vlasov type equations is derived, 
	describing the evolution of boson and fermion densities.
	In \cite{Mitrouskas2021}, 
	P. Pickl
	and the third author of this paper 
	consider 
	a scaling similar to \eqref{scaling}, 
	but the focus lies on the dynamics of a fixed number of bosons (called impurities) interacting with a dense Fermi gas. 
	While in \cite{Mitrouskas2021} the impurities experience an effective interaction akin to the one felt by the bosons in our model, the effective potential is not computed explicitly, nor is a connection established to large Bose-Fermi mixtures.
	Since the present work focuses on spectral properties rather than dynamics, both the problem and our approach differ substantially from those in \cite{Mitrouskas2021}.

	To our  best knowledge, this is the first work to investigate the low-energy spectrum of 
	large Bose-Fermi mixtures. 
	Our purpose is two-fold. 
	First, in Theorem \ref{thm1} 
	we rigorously justify the previously described picture by proving the emergence
	of an     \textit{effective} bosonic Hamiltonian  that incorporates
	fermion-mediated interactions. 
	Second, as an application, 
	in Theorem \ref{cor} we
	consider short-range interactions and 
	prove the existence of a stability-instability transition
	as detected experimentally in \cite{DeSalvo19}.\\[-10mm]

	\begin{center}
		\begin{figure}[t!]
			\includegraphics[scale=0.38]{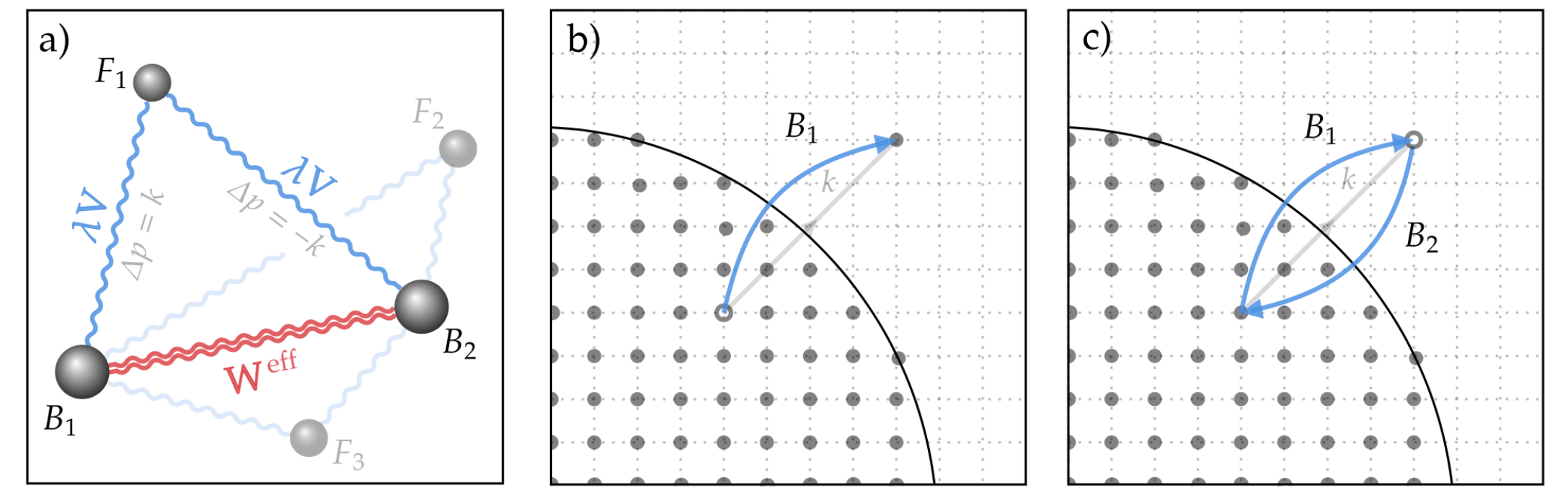}
			\caption{\hspace{-2.5mm}\label{fig1}\footnotesize{ Visualization of the effective interaction: (a) illustrates a coherent three-particle process (in blue) involving two bosons and one fermion, which generates the effective interaction (in red). (b) and (c) show the same process in momentum space, where the first boson excites a fermion from the Fermi ball and the second boson subsequently annihilates this excitation.}}
		\end{figure}
	\end{center}

	\subsection{Main results}\label{sec:main:results}
	In our first result, 
	we describe  the asymptotics of the low-energy eigenvalues of
	$H$ as $\kf  \rightarrow \infty$.
	In this limit, we prove
	the emergence of an effective    Hamiltonian 
	acting on the bosonic subsystem
	$\mathscr H _{ \rm B  }$ given by
	\begin{equation}
		\label{h:eff}
		h^{\eff}  
		=  
		\sum_{1 \leq i \leq N } (-\Delta_{x_i}) +  \frac{1}{N}
		\sum_{ 1 \leq i < j \leq N } \W^{\eff}(x_i - x_j) 
	\end{equation}
	where the effective two-body potential takes the form
	\begin{equation}
		\W^{\eff} = \W  - \,  V \ast V  \ .
	\end{equation}
	Here, $ (  f\ast g ) (x) = \int_{\To^3} f(y) g(x-y) dy$ denotes the convolution on $\To^3$. 
	%We observe that the second term is negative, as soon as $V$ has a definite sign, making the mediated interaction purely attractive for $V \ge 0$ or $V \le 0$. 
	We let 	the Fourier series be
	$\widehat V(k) = 
	(2\pi)^{-3/2} 
	\textstyle \int_{   \To^3 } V(x) e^{- ik \cdot x  } d x 
	$
	for any $ k \in \Z^3$ 
	and for  convenience    we put $ \hat V_0 \equiv \widehat V  (0)  $.
	For  $s \in \R$ we denote by   $H^s(\To^3)$    
	the Sobolev space on the torus. 
	
	\smallskip 
	
	In order to state our results, 
	we introduce the  precise condition for the potentials. 
	Throughout this work, we shall always assume:
	\begin{condition} 
		\label{Ass1}
		$W  \in L^p (\To^3)$   for some $ p \in ( \tfrac32  , \infty]$
		and $V \in H^4 (\To^3)$.
	\end{condition}
	
	Under these conditions, the   Hamiltonian $H$ 
	can be defined 
	as a self-adjoint  operator, understood as a sum of quadratic forms with form domain 
	$  	H^1(\To^{3N} )  \otimes 	H^1(   \To^{3M}	 ) \cap \mathscr H $.
	While this   is well-known, we include   a proof  of   relative form  boundedness 
	of the potentials in  Lemma \ref{lemmaA1} for completeness.  Notably,  the  Coulomb potential $ \widehat W(k)  \sim  |k|^{-2 }$ is included. \smallskip
	
	Finally, for    a lower bounded,  self-adjoint operator  $A$ 
	we denote by $  (  \mu_n( A)  )_{ n \geq1 }$ its
	sequence of eigenvalues counting multiplicities, in non-decreasing order. 
	Our first main result  now reads as follows.

	\begin{theorem} 
		\label{thm1} 
		Let $H$ be the Bose-Fermi  Hamiltonian  on $\mathscr H_B \otimes \mathscr H_F$
		given by 	\eqref{eq:H}, with scaling \eqref{scaling}, and let $h_\eff$ be the effective bosonic Hamiltonian on $\mathscr H_B$
		defined in \eqref{h:eff}. Assume $(W,V)$ satisfy Condition \ref{Ass1} for some $p\in (\tfrac32 , \infty ]$.
		Then, 
		there is a constant $C>0$  (depending only on the value of $p$) such that for any $\kf \geq 1 $, $ N \geq 1 $ and 
		$ 1 \leq n \leq  N $
		\begin{align*}
			\mu_n(H)  
			=
			E_{\rm F }  + 
			\lambda NM \hat V_0   
			-    \tfrac12 \, 
			\textstyle  \int_{\To^3 } |V|^2        dx 
			- 
			\tfrac12(N-2)  \, | \widehat V_0|^2
			+ 
			\mu_n (h_{}^{\eff })
			+ 
			\ve_{ N, \kf}
			%	\Big| \mu_n(H)    - E_{\rm F } - \lambda NM \hat V_0   +   \tfrac12 \, 
			%	{ \textstyle  \int_{\To^3 } |V(x)|^2 dx  }   
			%			 + N | \hat V_0|^2
			%			 - 
			%	 \mu_n (h_{}^{\eff }) \Big| 
			%	\leq 
			%	C  \Q   N^2    ( \ln \kf)^{5/3} \kf^{-1/3 }
		\end{align*}
		where the error satisfies
		\begin{equation}
			| 	 	 \ve_{ N, \kf}	 | 
			\, \leq  \, 
			C  \, 	  \big(  
			1+ \| \W  \|_{L^p}^{  2p/ (2p-3) }     + 
			\|   V \|_{H^4}^2 
			\big)^2  
			N^2    ( \ln \kf)^{5/3}   \kf^{ -1/3 }
		\end{equation}
		provided  $  \kf  (\ln  \kf )^{-5 } \ge C 
		(n   N 	\| V\|_{H^4}^2 	)^3$.
	\end{theorem}

	The statement involves three parameters: $n$, $N$ and $k_F$. To approximate the eigenvalues   of $H$ using those of $h^{\rm eff}$, the error must be  smaller than $\mu_n(h^{\rm eff})$.
	In particular, note that for fixed $n$ and $N$, we have $|\varepsilon_{N, \kf}| \to 0$ as $\kf \to \infty$. For example, for the ground state energy ($ n=1$), we obtain:
	\begin{align}\label{eq:H:n:convergence}
		\lim_{\kf \to \infty} 
		\big(  \inf \sigma (H) -  E_{\rm F } - \lambda NM  \hat V_0  \big) = 
		\inf \sigma (h_{}^{\eff })
		-
		\tfrac12    \,
		{ \textstyle  \int_{\To^3 } |V|^2 dx  }  
		-
		\tfrac12(N-2)  \, | \widehat V_0|^2
	\end{align}
	for all $N\ge 1$.  
	
	\begin{remarks}\phantom{a}
		
		\begin{itemize}[leftmargin=*]
			\item  
			The term 
			$
			- \tfrac12    \,
			{ \textstyle  \int_{\To^3 } |V|^2 dx  } 
			$
			represents the self-energy of the bosons, arising from a coherent two-particle process where a single  boson excites and subsequently annihilates the same fermion from the Fermi ball. 
			On the other hand, the term $\tfrac12 (N-2) |\hat{V}_0|^2$ must be subtracted in order to extract the explicit $V \ast V$ potential in the effective Hamiltonian. Both of these contributions are physically irrelevant, as they correspond to constant energy shifts in the effective Hamiltonian.
			\vspace{1mm}
			
			\item  
			Thanks to the explicit error bound in Theorem \ref{thm1}, both  potentials $V$ and $W$  can depend on the number of bosons $N$. This becomes relevant in Theorem \ref{cor}, where we consider rescaled, $N$-dependent  interactions.
			While,  in principle,  they can also depend mildly  on $k_F$, we    exclude this possibility since it is not relevant for our analysis.
		\end{itemize}
	\end{remarks}
	Before discussing our main application, we state a corollary on the convergence of the ground state eigenfunction of $H$. For simplicity, we focus on the ground state, though the result extends to other non-degenerate eigenfunctions. The corollary shows that the ground state is approximately of the form $\phi^{\rm eff} \otimes \Omega_{\rm F}$, where $\phi^{\rm eff}$ is the ground state of the effective Bose Hamiltonian $h^{\rm eff}$. This supports the heuristics discussed above on the effective decoupling between the Fermi and Bose subsystems. The proof is postponed to Appendix \ref{appendix:projections}.

	\begin{corollary}\label{coro} 
		Let $\Phi \in \mathscr H$ and $\phi^{\rm eff}\in \mathscr H_{\rm B}$ be the normalized ground state eigenfunctions of $H$ and $h^{\rm eff}$, and fix their relative phase so that $\langle \Phi, \phi^\eff \otimes \Omega_{\rm F} \rangle \ge 0$. Then, for all $(W,V)$ satisfying Condition \ref{Ass1} and $N\ge 1$, we have
		\begin{align*}
			\lim_{k_{\rm F}\to \infty} \big\|\, \Phi \, - \,  \phi^{\rm eff} \otimes \Omega_{\rm F} \, \big\| = 0
		\end{align*}
	\end{corollary}

	%\comm{TBD}
	%The next relevant implication of Theorem \ref{thm1} 
	%is that  of \textit{enhanced binding}.
	%This effect occurs  when the coupling between bosons and
	%    fermions strengthens the attractive interaction among bosons, possibly leading to the emergence 
	%    of a  localized bound state. 
	%    For instance, consider   $N=2$ bosons  with 
	%$\W \equiv 0$ and    $V(x )=g \, v(|x|)$ for some  non-negative, radial function $v(r)$
	%and  a coupling $g \in \R$. 
	%In this situation, the fermions induce an attraction between the two bosons that---for large enough values of   $|g|$---leads to the formation of a 
	%negative eigenvalue of the Schrödinger operator $h^{\eff}(r) = -\Delta_r -  \frac12 g^2 v\ast v(r)$ with $r$ the relative coordinate between the bosons.
	%\smallskip
	
	Our main application of Theorem \ref{thm1} addresses the stability-instability transition in large Bose-Fermi mixtures, as mentioned in the introduction. The physical picture is that the bosons form a BEC immersed in a large fermionic environment \cite{DeSalvo19}. For sufficiently weak coupling, the condensate remains stable under the interaction with the Fermi gas, and the energy per particle of the bosonic subsystem stays finite in the limit $N \to \infty$. In contrast, when the Bose-Fermi coupling becomes strong enough, the effective boson-boson potential turns sufficiently attractive, driving the bosonic subsystem into an unstable phase. In this regime, the bosonic energy per particle is no longer uniformly bounded below, which indicates a collapse of the condensate.
	
	To state a precise version of this transition, we consider the Gross-Pitaevskii (GP) regime that 
	describes a dilute gas of ultracold atoms interacting via short-range potentials. Our precise assumptions are as follows.
	
	\begin{condition}
		\label{Ass2}
		There exist 	$	w \in C(\R^3)$ and  non-zero $ v \in H^4 (\R^3)$, both non-negative, radial and compactly supported, 
		such that 
		\begin{equation}
			W(x)= N^3 w (Nx )  \qquad \t{and} \qquad 
			V(x) = g \, N^3 v(Nx)  \ , \qquad N \geq 1
		\end{equation}
		with  coupling constant $g \in \R. $
	\end{condition}
	
	Under this condition, the Hamiltonian $H$ describes a Bose-Fermi mixture with repulsive Bose-Bose potential and repulsive 
	($ g\ge 0$) or attractive ($ g\le 0$) Bose-Fermi interaction, with both $W$ and $V$ approximating delta  distributions as $N\to \infty$. We note that $  (V\ast V) (x) = g^2 N^3  ( v\ast v )  (Nx)$ and thus the effective  re-scaled 
	two-body potential is 
	\begin{align}
		w_g  \,  \equiv  \, w   \, -  \,    g^2 \,  ( v*v ) .
	\end{align}
	The effective Hamiltonian \eqref{h:eff}  on $\mathscr H_B$
	becomes
	\begin{equation}
		\label{eq:heff}
		h^\eff = \sum_{1 \leq i \leq N }  ( - \Delta_{x_i })
		+ \frac{1}{N}
		\sum_{1 \leq i < j \leq N }
		N^3 w_g ( N  (x_i -x_j)) \ ,
	\end{equation} 
	which is commonly referred to as GP-scaling for $N$ bosons. In Theorem \ref{cor} below, we  prove  the
	existence of a  \textit{critical coupling} $g_c\ge 0$
	from which we can distinguish two regimes: First, the \textit{stable regime}, 
	given by  $   0\le  | g |  <  g_c$.
	Here,  the    ground state energy per particle 
	of the 
	effective   Hamiltonian $h^{\eff}$
	is uniformly bounded    in   $ N     $.
	Second,  the \textit{unstable regime}, given by $ | g | >g_c$, where the pair potential $w_{g}$ becomes sufficiently attractive. Here, the ground state energy  per particle diverges to $- \infty $ as $N\to \infty$.
	
	\smallskip
	
	Within the stable regime, it would be desirable to calculate
	the energy per particle.
	To this end, we identify 
	yet another parameter region. Namely, where   $ | g |  $ is small enough so that  $w_g \ge 0$. 
	In this case, it is well known
	that  bosons form a BEC with ground state energy
	\begin{align}\label{eq:stable:h:n:energies}
		\inf \sigma ( h^\eff) = 4 \pi \mathfrak a(g) N +  O(1) \quad \text{as} \quad N\to \infty \ ,
	\end{align}
	where  $\mathfrak a(g)$ denotes the scattering length associated with the   potential $w_{g}$.
	That is, 
	\begin{equation}
		\label{eq:a}
		\mathfrak a(g)   =  \frac{1}{8\pi} \int_{\R^3}	 w_{g} (x) f_g(x) d x  \qquad
		\t{where}
		\qquad
		\begin{cases} 
			& 	\hspace{1.2mm} 	( -\Delta  + \tfrac12  w_g )  f_g = 0		\\[1mm]
			&  	\lim_{|x | \rightarrow \infty } |f_g(x)|=1
		\end{cases} \ . 
	\end{equation}
	We refer to \cite{LSY2000} and \cite[Theorem 1.1]{BBCS2019} for precise formulations and proofs of \eqref{eq:stable:h:n:energies}. Additionally, let us note that the analysis of dilute Bose gases has been a very active area of research in mathematical physics---for instance \cite{Lieb:book,NRS2016,BBCS2019,
		BBCS2020,BBCS2020b,
		Fournais2020,Fournais2023,Brooks}---with many results extending beyond \eqref{eq:stable:h:n:energies}.  While the condition $w_g \geq 0$ is likely not necessary for Bose-Einstein condensation, the exact necessary and sufficient conditions are, to our knowledge, still unknown; see \cite{Baumgartner,Yin2010} for partial results. Lastly, let us mention that it is also an open question whether alternative stable phases exist besides BEC. These problems are of interest beyond the context of Bose-Fermi mixtures and outside the scope of this work.
	
	\smallskip
	
	The  stability-instability transition is summarized in the following theorem.  A heuristic illustration of the statement is shown in Figure \ref{fig2}.

	\begin{theorem} 
		\label{cor} Assume  that Condition \ref{Ass2} holds with coupling $g \in \R$, and
		define the shifted  total energy-per-boson 
		$  \mathcal E   : \R \to  \R \cup \{  -  \infty\} $  as
		\begin{align}
			\label{def:E}
			\mathcal E(g)    \equiv   	\liminf_{N\to \infty}    \,
			\lim_{k_F\to \infty}    \frac{1}{N} \, 
			\Big( 
			\inf \sigma \big( H )     - E_F - \lambda NM  \hat V_0    +  
			\tfrac12    \,
			{ \textstyle  \int_{\To^3 } |V(x)|^2 dx  }   
			+ 
			\tfrac12(N-2)  \, | \hat V_0|^2
			\Big)   
			\ .
		\end{align}
		Then, the following statements hold
		\begin{enumerate} 
			\item  
			$\mathcal E(g) = \mathcal E(-g)$ for all $g\in \R$ and $  g \in [0 , \infty) \mapsto \mathcal E(g) $ is a monotone decreasing map. 
			\item 
			There exists a critical coupling $ 0 \leq g_c < \infty$
			such that 
			\begin{equation}
				\mathcal E(g)  > - \infty   \qquad \forall g\in \{0\}\cup (0, g_c)
				\qquad
				\t{and}
				\qquad 
				\mathcal E(g) = - \infty \qquad \forall g > g_c \ . 
			\end{equation}
			\item 
			Let $g_0  \equiv  \sup  \{ g \geq 0  : w - g^2 (v*v) \geq 0 \}$ and $g_*  \equiv w(0) / \| v\|_{L^2}^2$.
			Then,  $g_0 \leq g_c \le g_*$.
			In addition,  
			for all   $ 0 \leq  g  \leq g_0 $ the limit $N\to \infty$  in 
			\eqref{def:E}
			exists, and 
			\begin{align}
				\label{Eg}
				\mathcal E  (    g ) = 4\pi \mathfrak a (g)   \  ,
			\end{align}
			where the right hand side is defined in terms
			of the scattering length \eqref{eq:a}. 
		\end{enumerate}
		
		%Moreover, there  exists a critical  coupling 
		%	$   g_c \in [ g_* , \infty)$  such that 
		%		\begin{align}\label{eq:thm:transition}
			%			  \mathcal E(g) 
			% = 
			%  				 \begin{cases} 
				%					4\pi \mathfrak a(g)  \, ,   &  
				% g \in \{ 0 \} \cup (  0 , g_0]
				%						 \\[1mm]
				%						-\infty   \, ,    &   
				%				 g    \in (    g_c   , \infty) 
				%					\end{cases} 
			%			\end{align}
		%			where  $g_0  \equiv  \sup  \{ g \geq 0  : w_g \geq 0 \}$ and $g_*  \equiv  \inf \{ g \ge 0 : w_g(0) \le 0\} = w(0) / \| v\|_{L^2}^2$. 
	\end{theorem}

	\begin{center}
		\begin{figure}[t!]
			\vspace{4mm}
			{\hspace{-5mm} 
				\includegraphics[scale=0.475]{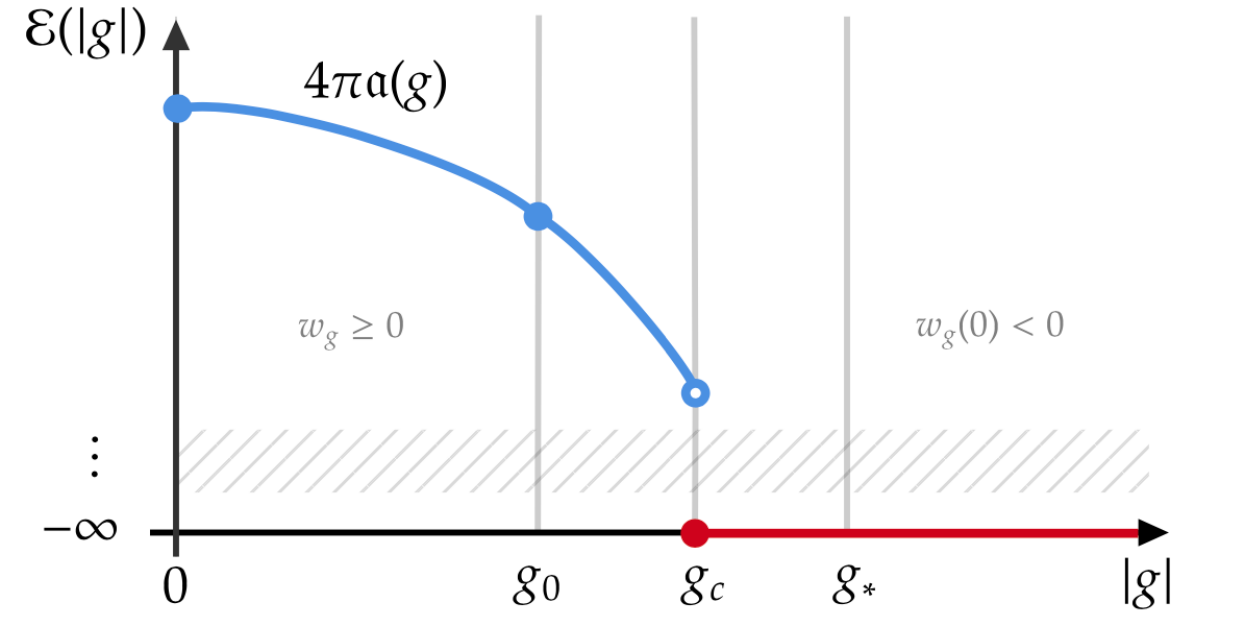}
			}
			\caption{\hspace{-2.5mm}\label{fig2} Sketch of the phase diagram: For $w_g\ge 0$, the energy satisfies $\mathcal E(g) = 4\pi \mathfrak a(g)$. In the intermediate stable regime $g_0\le g \le g_c $, only monotonicity of $\mathcal E(g)$ is known. In the unstable regime, $\mathcal E(g)=-\infty$.}
		\end{figure}
	\end{center}  
	
	\begin{remarks} 
		The following points are worth noting.
		\begin{itemize}[leftmargin=*]
			
			\item 
			The contribution due   to  self-energy   scales like $\| V \|_{L^2}^2 = g^2 N^3 \| v \|_{L^2}^2$, and thus it must be subtracted from $\inf \sigma(H)$  in order
			to detect  the stability-instability transition.
			
			\vspace{1mm}
			
			\item    In general, we cannot rule out the case $g_c = 0$, meaning the transition might occur as soon as $|g| > 0$. However, if  $w_g \ge 0$ for small $g$---for example, by requiring $\supp(v * v) \subset \supp(w)$---then $g_c \ge g_0 > 0$. This ensures  a non-trivial stable regime.
			
			\vspace{1mm}
			
			%	\item  
			%		In the regime $0\le |g| \le g_0$, that is, for $\{g\in \R: w_g\ge 0\}$, 
			%		one can also consider 
			%		for   $ n \geq 1 $ 
			%		\begin{equation}
				%			\mu_n(h^{\rm eff}) = 4 \pi \mathfrak a(g) N + m_n(g) + o(1) \quad \text{as}\quad N\to \infty \, ,
				%		\end{equation}
			%		where the sequence    $( m _n (g))_{ n =1}^\infty$
			%		corresponds to the eigenvalues of an %explicit 
			%		Bogoliubov-type Hamiltonian on Fock space.  
			%		For the precise statement of the next-order expansion in the GP-model, we refer to \cite{BBCS2019}. 
			
			%		\vspace{1mm}		
			
			\item  	  
			It is possible that $g_c = g_0 = g_*$,  i.e.,
			the intermediate phase is empty.
			From our proof, this follows for instance if 
			$w_{g_0}(0)=0$.   \textit{Example:} Choose  $w = \alpha  (  v*v) $ for some $\alpha>0$, 
			so that $g_c = g_0 = g_* = \alpha$. 
			Additionally,  in this case $  \mathcal E (g)$ is 
			left-continuous at $g_c$, and $ \mathcal E (g_c)=0$. 
			
			\vspace{1mm}

			\item  
			In addition to the GP-scaling, one may also
			consider less singular re-scaled potentials
			\begin{align} 
				\W (x) = N^{3 \beta }  w (N^\beta x)\quad \text{and} \quad 
				V(x) = g N^{3 \beta }  v(N^\beta x ) \ 
			\end{align}
			for  $0  \leq  \beta < 1 $. The Bose gas energy for such models has been studied in \cite{Seiringer2011,LNR14,BBCS2020b,
				BLPR2023}. With these scalings, 
			the analogous formula to \eqref{Eg} for $g \in \{ 0\} \cup [ 0,g_0 ]$ is
			\begin{equation}
				\label{Eg2}
				\mathcal E (g) = 4\pi 	\(\textstyle 
				\int_{\R^3} w(x) d x 
				-      g^2  
				\Big(   \int_{\R^3} v(x) d x  
				\Big)^2 \)    , 
			\end{equation}
			which comes from replacing the scattering length $\mathfrak{a }(g)$
			with $\int_{\R^3} w_g (x) d x $. Let us emphasize that similar formulas to \eqref{Eg2} have been derived in the physics literature; see   \cite{Tsurumi2000,Chui2004,DeSalvo19}.

			%	Similar formulas  to \eqref{Eg2} can be found in the physics literature.
			%For instance, 	in \cite{DeSalvo19} (see also  in  \cite{Tsurumi2000,Chui2004}) the authors
			%	consider an inhomogeneous Bose-Fermi mixture
			%	and introduce the Gross-Pitaevskii energy functional for the bosonic subsystem
			%\begin{equation}
			%	\label{E}
			%	 E (\psi) =
			%	 \int_{\R^d}
			%	 \(  
			%	  \frac{\hbar^2}{2m_B}
			%	 |\nabla \psi (x)|^2 
			%+ V_{\t{eff}} (x) 	 |\psi (x)|^2 + 
			% \frac{g_{\eff} }{2}
			% | \psi (x)	|^4 
			% \) d x
			%  \end{equation}
		%where 
		%$V_{\eff} (x)$ 
		%is an effective trapping potential, 
		%and 
		%\begin{equation}
		%	\label{geff}
		%	 g_{\eff}
		%	   = 
		%	   g_{\t{BB}}  
		%	    - 
		%	    \frac{ 3 }{2}
		%	    \frac{ M }{E_F}
		%	    g_{\t{BF}}^2 
		%\end{equation}
		%is an effective coupling constant, 
		%determined from the bare couplings 
		%$g_{\t{BB}}$ and $g_{\t{BF}}$.  
		%In the homogeneous case, i.e., when $V_\eff \equiv 0$, and considered on the unit torus, the ground state energy of \eqref{E} for $g_\eff \ge 0$ equals $\inf (E) = 2\pi g_\eff$, which resembles the form of \eqref{Eg2} in our model.
	\end{itemize}

\end{remarks}

\subsection{Emergence of the effective theory}\label{subsec:eff}

In this section, we sketch the proof of Theorem \ref{thm1} and place the  main  ideas in  a broader context. The key observation is that the bosons form a \textit{slow} subsystem, while the fermions act as the \textit{fast} subsystem. 
Here, the typical kinetic energy of a fermion is much larger than that of a boson, and
one speaks of a separation of energy scales. 
From this perspective, our model belongs to a larger class of physical systems exhibiting an adiabatic decoupling between two coupled subsystems. The analysis of such systems has a long-standing history, dating back to the work of Born and Oppenheimer \cite{BornOppenheimer}. A fundamental consequence of the separation of energy scales is that the two subsystems effectively decouple, with the fast subsystem inducing an effective interaction in the slow subsystem. This phenomenon has been rigorously established in various contexts, including Born--Oppenheimer theory for heavy nuclei \cite{PST,TeufelSpo,Stefan}, the strongly coupled Fröhlich polaron \cite{Mit,Leopold,BrooksMit,BrooksSei}, and different Nelson-type quantum field theory models \cite{Davies79,Hiroshima99,TeufelNelson,Tenuta,CarMit24b}.

In Theorem \ref{thm1}, we verify this effect for large Bose--Fermi mixtures, where the ratio of energy scales is quantified by the Fermi momentum $ \kf  \gg 1$. 
While in our Hamiltonian the dependence of various terms with respect to $\kf$ is not explicit, the underlying scaling can be   understood  in terms of a simplified \textit{abstract} model, presented below. 
Here,  we sketch the proof for this abstract model, as we believe that our ideas 
apply not only  to   Bose--Fermi mixtures,  but to a wide range of physical systems. 
As we explain below, 
one of the key estimates in our proof  is an operator inequality  
based on   the  completion of the square in the operator sense. 
It allows us to obtain the lower bound for the ground state energy, as well as the higher order eigenvalues. 
The abstract model that we present includes an explicit scaling parameter $\mu \gg 1$ to quantify the separation of energy scales; this choice of scaling is reminiscent of the one in the models studied in \cite{BrooksMit,Davies79,Hiroshima99,CarMit24b}.  We then  discuss the modifications required to apply the argument to the Bose--Fermi model.

%
%One of the challenges is that the Fermi momentum does not appear as an explicit scaling parameter in the Hamiltonian. Therefore, to better illustrate the main idea—particularly the emergence of the effective interaction—we first present the argument for a simplified \textit{abstract} model. Here, we introduce an explicit scaling parameter $\mu \gg 1$ to quantify the separation of energy scales. This choice of scaling is reminiscent of the one in the models studied in \cite{BrooksMit,Davies79,Hiroshima99,CarMit24b}. Subsequently, we discuss the modifications required to apply the argument to the Bose-Fermi model.

\subsubsection{
	The abstract model}
Consider a Hilbert space $ \mathscr{H}_A \otimes \mathscr{H}_B $, where $ \mathscr{H}_A $ and $ \mathscr{H}_B $ represent a fast and a slow subsystem, respectively. We assume a Hamiltonian of the form 
\begin{align}
	\mathbb{H} =  \mu T_A  \otimes \1  +  \1 \otimes  T_B   + \mu^{1/2} F, 
\end{align}
where $ T_A \geq 0 $ and $ T_B \geq 0 $ are the Hamiltonians of the fast and slow subsystems, respectively, and $ F $ models their interaction. The parameter $ \mu \gg 1 $ controls the ratio of energy scales.  We   assume that $ T_A $ has a unique ground state $ \Phi_A $ with eigenvalue $ \inf \sigma(T_A) = 0 $, a spectral gap of order one, and that $ P F P = 0 $ with spectral projection $P = | \Phi_A \rangle \langle \Phi_A | \otimes \1 $.

We divide the analysis into upper and lower bounds, and sketch it only for the ground state energy. For the upper bound, we use a trial state of the form
\begin{align} 
	\Psi = ( \1 - \mu^{-1/2} T_A^{-1} F) \Phi_A \otimes \Phi_B^{\rm eff}, \quad  \Phi_B^{\rm eff} \in \mathscr{H}_B. 
\end{align}
Under appropriate assumptions on the interaction (e.g, if $F$ is bounded), 
one proves that 
\begin{align}
	\langle \Psi, \mathbb{H} \Psi \rangle = \langle \Phi_B^{\rm eff}, h_B^{\rm eff} \Phi_B^{\rm eff} \rangle + O(\mu^{-1/2}) \quad \text{with} \quad  h_B^{\rm eff} = T_B - \langle \Phi_A, F T_A^{-1} F \Phi_A \rangle_{\mathscr{H}_A}.
\end{align}
The effective Hamiltonian $h_B^\eff$
contains a non-positive
operator on $\mathscr H_B$, 
representing 
a partial expectation value with respect to $ \Phi_A \in \mathscr{H}_A $.
More precisely, for an operator $ S $ on $ \mathscr{H}_A \otimes \mathscr{H}_B $ and fixed $ \Phi_A \in \mathscr{H}_A $, the sesquilinear form
\begin{align}
	\< \Phi_B ', \< \Phi_A , S \Phi_A \>_{\mathscr{H}_A} \Phi_B \>_{\mathscr{H}_B} \equiv \< \Phi_A \otimes \Phi_B', S \Phi_A \otimes \Phi_B \>_{\mathscr{H}_A \otimes \mathscr{H}_B}
\end{align}
defines an operator $ \Phi_B \mapsto \< \Phi_A, S \Phi_A \>_{\mathscr{H}_A} \Phi_B $, which we call the partial expectation value of $ S $ with respect to the state $ \Phi_A \in \mathscr{H}_A $.  
In our analysis, we take $S = F T_{A}^{-1 }F $. 
We then choose $ \Phi_B^{\rm eff} $ as the ground state of $ h^{\rm eff}_B $ and apply the variational principle  to    obtain   $ \inf \sigma(\mathbb{H}) \le \inf \sigma(h^{\rm eff}_B) + O(\mu^{-1/2}) $. For our specific model, this analysis is carried out in Section \ref{sec:upper:bound}.

To see how the effective Hamiltonian emerges as a lower bound, we  \textit{complete the square} in the following form
(we omit the tensor product with the identity)
\begin{align} \label{completing:square:abstract}
	\mathbb H  = \mu \big( T_A^{1/2} + \mu^{-1} T_A^{-1/2} F \big)^* \big( T_A^{1/2} + \mu^{-1} T_A^{-1/2} F \big) + T_B - F T_A^{-1} F \,  . 
\end{align}
This leads to the operator inequality $ \mathbb H  \ge T_B - F T_A^{-1} F$. Assuming we can show for suitable low-energy states $\Psi  \in \mathscr H$ (e.g. the ground state of $\mathbb H$) that
\begin{align} \label{FTF:bound}
	\langle \Psi, F T_A^{-1} F \Psi\rangle \le \langle \Psi, \1   \otimes \langle \Phi_A, F T_A^{-1} F \Phi_A \rangle_{\mathscr H _A} \Psi \rangle + o(1) \quad \text{as} \quad \mu \to \infty,
\end{align}
we obtain the matching lower bound $ \inf \sigma(\mathbb H) \ge \inf \sigma(h^{\rm eff}_B) + o(1)$. This bound is non-trivial and depends on specific model properties---in our setting, $\Phi_A$ corresponds to the Fermi sea and \eqref{FTF:bound} emerges from normal ordering of creation and annihilation operators in combination with suitable operator inequalities. We refer to Section \ref{sec:lower:bound} for the details.

\subsubsection{Application to Bose-Fermi mixtures}

To apply the outlined arguments to the Bose-Fermi model, we first note that bosons mostly interact with fermions by creating or annihilating particle-hole pairs in the Fermi ball. Therefore, we introduce the so-called particle-hole transformation, which implements this directly at the level of the Hamiltonian. Up to unitary equivalence, we then study a new Hamiltonian $\mathbb H$---defined in Eq. \eqref{def:hamiltonian:exc}---that describes the interaction between bosons and \textit{excitations} of the Fermi ball. Since most particle-hole excitations have large momentum $|p| \sim k_F$, we identify them as the fast subsystem relative to the slow bosons with energies of order one. While the excitation Hamiltonian $\mathbb H$ retains the key assumptions of the abstract  model discussed above, it introduces   additional difficulties:
\begin{itemize}[leftmargin=7mm]
	\item[\bf{1.}] The Fermi momentum $k_F$ plays the role of $\mu$, but it does not appear as simple prefactors in the Hamiltonian. Thus, the scaling of the different terms must be carefully extracted.\\[-3.5mm]
	\item[\bf{2.}] All operators are unbounded and require careful estimates to control error terms. \\[-3.5mm]
	\item[\bf{3.}] The explicit form of the effective Hamiltonian  requires 
	analyzing the partial expectation value
	$\langle \Phi_A, F T_A^{-1} F \Phi_A \rangle$. 
	In the Bose-Fermi model this operator still depends on  $k_F$, and to obtain its explicit form---given by $\frac1N \sum_{1\le i < j \le N } (V\ast V)(x_i-x_j) $---requires the analysis of its large $k_F$ limit.
\end{itemize}

	\subsection{The small mass limit}
	
	As explained in the introduction, the present work is largely motivated by the experimental findings in~\cite{DeSalvo19} on Bose–Fermi mixtures consisting of light fermionic $^6\mathrm{Li}$ atoms and heavier bosonic $^{133}\mathrm{Cs}$ atoms. We denote the respective masses by $m_{\mathrm{F}}$ and $m_{\mathrm{B}}$. In these experiments, the separation of energy scales arises due to the small mass ratio $m \equiv \frac{m_{\mathrm{F}}}{m_{\mathrm{B}}} \ll 1$.
	While in our analysis we set $m=1$ and take the limit $k_F \rightarrow \infty$, it is also possible to consider a vanishing mas ratio $m \rightarrow 0$ at \textit{fixed} $k_F$. The analysis of this  limit is well within the scope of our methods---let us briefly sketch the necessary adjustments. 
	
	First, let $H$ describe such a mixture in units where $m_B =1$ and  $ m = m_{\rm F } \ll1 $.
	The Fermi energy is now given by $E_{\rm F } = m^{-1 } \sum_{ k \in B_{\rm F}} k^2 = O( m^{-1 } )$ and we modify the coupling   to $\lambda = 1 / \sqrt{4 \pi N m k_F}$.
	Following the    proof of Theorem \ref{thm1},
	we  would then obtain the  following  result, analogous to \eqref{eq:H:n:convergence}:
	\begin{equation} 
		\lim_{m \rightarrow 0 } \big( \inf \sigma (H) - E_{\rm F } - \lambda NM \hat V_0 \big) = \inf \sigma (h_{ \kf }^{\eff }) - \tfrac12 \, W_{k_F}(0) \, , 
	\end{equation} 
	where the effective Hamiltonian $h_{\kf}^\eff$ and    potential $W_{\kf}$ are defined as  in \eqref{heff} and \eqref{eq:def:eff:kF:potential}, respectively. Notably, the effective potential now depends on $k_F$, and it agrees with $V * V$ only in the limit $k_F \to \infty$.

	\smallskip 
	
	\noindent  \textbf{{Organization of the proof}}.
	In Section \ref{sec:sec:quant}
	we introduce the second quantization formalism
	and reduce the analysis to the excitation Hamiltonian $\mathbb H$ on Fock space.
	In Section \ref{sec:proofs}
	we prove Theorem \ref{thm1} and \ref{cor}.
	The proof is based on 
	corresponding upper and lower bounds for the eigenvalues of $\mathbb H$.
	In Section \ref{sec:prelims} we collect helpful estimates for its analysis; 
	these involve sum over relevant geometric sets, and various operator inequalities. 
	In Section \ref{sec:effective}
	we prove these bounds
	and show the emergence of the effective Hamiltonian. 
	Finally 
	in Section \ref{section:sums} 
	we prove an inequality first stated in Section \ref{sec:prelims}. 
	
	\medskip

	\noindent  \textbf{Convention}. We say that $C>0$  is a \textit{constant} 
	if it is a positive number, independent of all the physical parameters $N, M, \lambda$
	and also of the potentials $V$ and $W$. 
	Its value may change from line to line.

	\section{Second quantization}
	\label{sec:sec:quant}
	As we have mentioned earlier, the effective Bose-Bose interaction that is induced by the Fermi subsystem 
	can be understood as exchanging (virtual) excitations of the Fermi ball.                                       
	In order to make this idea precise, it is convenient
	to work in the formalism of second quantization, 
	where we allow the number of fermions to fluctuate.
	
	To this end, we  embed the fermionic Hilbert space 
	$\mathscr H _{F}$  into the larger  Fock space 
	\begin{equation}
		\mathscr F  =  \textstyle \C \oplus \bigoplus_{ n=1}^\infty \mathscr F^{(n)} 
		\qquad 
		\t{with}
		\qquad 
		\mathscr F^{(n)} = \bigwedge_{j=1}^n L^2 ( \To^3 	)  \ ,   \quad   n \geq 1    
	\end{equation}
	so that  $\mathscr H _{F}  =  \F ^{ (M)}	$. 
	As usual, $\F$
	is naturally endowed with an inner product
	\begin{align}
		\<  \Psi, \Phi \>_\F = \sum_{n = 0}^\infty \<  \Psi^{(n)} , \Phi^{(n)} \>_{\F^{ (n  )  }  } \ , \qquad \Psi, \Phi \in \mathscr F 
	\end{align}
	and we denote the Fock vaccum by 
	$\Omega   = ( 1 ,  \textbf{0})$. 
	Additionally,  we equip $\F$  with 
	creation and annihilation operators, 
	defined for $f \in  L^2 ( \To^3)$  and $\Psi \in \F$
	\begin{align}
		[  	  a(f) \Psi  ]^{(n)}  (x_1 , \cdots, x_n ) 
		& \ =  \ 
		\sqrt { n  +1 } \int_{\To^3} \overline{f(x)} \Psi^{ (n + 1 )}(x_1 , \cdots, x_n) d x  \\
		[  	  a^*(f) \Psi  ]^{(n)}  (x_1 , \cdots, x_n )
		& \  =  \ 
		\frac{1}{\sqrt  n}
		\sum_{ i = 1}^n (- 1 )^i  f(x_i ) \Psi^{ ( n -1 )} ( x_1, \cdots, \hat x_{i} , \cdots, x_n )  \ . 
	\end{align}
	where $\hat x_i$ is ommitted, and $n \geq 0$. 
	In particular, denoting by  $\{ A,B \}= AB +BA$  the anticommutator,
	they satisfy the Canonical Anticommutation Relations (CAR)
	\begin{align}
		\{   a(f), a^*(g) 			\}  = \< f,g \>_{L^2(\To^3)}  \ , 
		\qquad 
		\{  a^\flat (f), a^\flat (g)   \} =  0 \  ,
	\end{align}
	where $f,g\in L^2(\To^3)$, and $a^\flat(f)$ stands for either $a(f)$ or $a^*(f)$.  It is well-known that  from these relations   both operators
	become bounded and satisfy 
	\begin{equation}
		\| a (f) \| 	=		 \| a^* (f) \|  =  \| f \|_{L^2 }  \  .
	\end{equation}
	We now pass from the position to the momentum representation. 
	To this end, we recall $e_k(x)$ corresponds to the plane-wave basis (see \eqref{plane-wave}), and we define  for any $ k \in \Z^3$
	\begin{equation}
		a_k = a (	 e_k 	)  \  ,  \qquad a_k^*    = a^* (e_k) \ . 
	\end{equation}
	The 
	CAR takes  now reads$\{  a_k , a_\ell^*  \} = \delta_{k , \ell}$
	and $\{  a_k^\flat, a_\ell^\flat\}=0$.
	Note    also   $\| a_k \| = \| a_k^* \|  =1 \ . $
	We   let  $\mathcal N = \sum_{ k \in \Z^3} a_k^* a_k$
	denote  the standard number operator.

	The original Hilbert space $\mathscr H = \mathscr H_{B} \otimes \mathscr H _{F}$
	is then  naturally embedded into 
	$\mathscr H _{B} \otimes \mathscr F $ via the inclusion map. 
	The many-body 
	operator   \eqref{eq:H}  can now be expressed
	as the restriction  $ H = \mathcal H  \upharpoonright \mathscr H  $
	of  the Hamiltonian  
	\begin{equation}
		\label{H}
		\mathcal H = 	 
		h \otimes \1 
		+ \1 \otimes 
		\sum_{ k \in \Z^3}   k^2   \, a_k^* a_k 
		+ 
		\lambda 
		\sum_{1 \leq i \leq N }
		\sum_{ k ,\ell \in \Z^3 } 
		\widehat V(	 \ell - k 		) e^{  - i (\ell - k  )  x_i}  \otimes a_\ell^*  a_{k  }  
	\end{equation}
	defined   on $\mathscr H_B \otimes \F  $.

	\subsection{The particle-hole transformation}
	\label{section:particle-hole}
	In order to extract the relevant terms  for the eigenvalues $\mu_n(H)$
	we introduce 
	a change of variables
	which is referred to 
	as the 
	particle-hole transformation.
	This correspond to the unitary map 
	\begin{align}
		\mathscr R : \mathscr F \rightarrow \mathscr F
	\end{align}
	defined 
	via  its action on creation and annihilation operators
	as 
	\begin{align}
		\mathscr R^* a_k \mathscr R = \chi^\perp(k) a_k + \chi(k) a_k^*
		\qquad 
		\t{and}
		\qquad 
		\mathscr R^* a_k^*  \mathscr R = \chi^\perp(k) a^*_k + \chi(k) a_k 
	\end{align}
	and set the image of the Fermi ball
	as the new vacuum $\mathscr R \Omega_F    \equiv  
	\Omega $. 
	Here and in the sequel we use the notation for the 
	characteristic functions for $ k \in \Z^3 $
	\begin{equation}
		\chi(k)  \equiv  \1 ( | k | \leq k_F )
		\qquad
		\t{and}
		\qquad
		\chi^\perp(k)    \equiv   \1 ( | k | > k_F ) \ . 
	\end{equation}
	The map $\mathscr R $ is trivially extended to 
	$ \mathscr  H_{B } \otimes \F$ by tensoring with the identity. 
	Finally, we introduce the following 
	new set of creation and annihilation operators
	for \textit{particles} and \textit{holes}, respectively 
	\begin{align}\label{def:bk:ck}
		b_k \equiv  \chi^\perp(k) a_k			
		\qquad
		\t{and}
		\qquad 
		c_k  \equiv \chi(k) a_k. 
	\end{align}
	They satisfy the following version of the CAR
	\begin{align}
		\{ b_k ,b_\ell^* \} = \,  \chi^\perp (k) \,  \delta_{ k ,\ell}  \,      \ , 
		\qquad 
		\{	 c_k, c_\ell^*		\} =   \, \chi (k)  \,  \delta_{ k ,\ell}  
		\quad 
		\t{and}
		\quad 
		\{ 	 b_k^\flat, c_\ell^\flat\} =  0  
	\end{align}
	for all $ k , \ell \in \Z^3$. 
	The main advantage of introducing the
	new  set  of operators 
	$\{ b_k \}$ and $\{ c_\ell \}$ 
	is that they carry all   the information 
	regarding momentum restrictions.

	Let us   now apply the particle-hole transformation
	to   the Hamiltonian $ \mathcal H $. 
	A straightforward calculation yields
	\begin{align}
		\label{conjugate}
		\mathscr R^* \mathcal H  \mathscr R 
		= 
		E_F + \lambda N M  \hat V_0  +\mathbb H  \ .
	\end{align}
	Here the first term corresponds to the kinetic energy of the Fermi ball, 
	whereas the second term corresponds to an additional energy shift 
	due to its interaction with the bosons (a constant term due to translation invariance).
	On the other hand, the operator $\mathbb H$
	corresponds to the Hamiltonian
	describing the Bose gas and its interactions with the \textit{excitations} of the Fermi ball. 
	It takes the form
	\begin{equation}
		\label{def:hamiltonian:exc}
		\H  = h \otimes \1 + \1 \otimes \T + \lambda \V  \  . 
	\end{equation}
	The first term
	is      the Bose Hamiltonian \eqref{eq:HB}. The second term describes the kinetic energy of the excitations, and is given by 
	\begin{equation}
		\label{eq:T}
		\mathbb T =    
		\sum_{ k\in \Z^3 } 
		\Big(	  k^2 b_k^*b_k - k^2 c_k^* c_k			\Big)	 \  .
	\end{equation}
	The third term describes    interspecies  interactions
	and  admits  the   decomposition 
	\begin{align}
		\V  =  \V_+ + \V_-  +  \V^{\t{diag}}   \ , 
	\end{align}
	where each contribution is given by 
	\begin{align}
		\label{V+}
		\V_+  
		&    \, =  \, 
		\sum_{ i = 1 }^N 
		\sum_{p , k \in \Z^3 }
		\widehat V (k   )
		e^{ - i  k  \cdot x_i}
		\otimes  
		b_{p}^* c_{p + k  }^*  \\
		\label{V-}
		\V_ -  
		&  \, =  \, 
		\sum_{ i = 1 }^N 
		\sum_{p , k \in \Z^3  }
		\widehat V ( k   )
		e^{ +  i  k   \cdot x_i}
		\otimes    c_{p+k }   b_{p } \\
		\V^{\rm diag}  
		&  \, =   \, 
		\sum_{ i = 1 }^N 
		\sum_{\ell  , k \in \Z^3 }
		\widehat V ( \ell - k )
		e^{ - i ( \ell - k  ) \cdot x_i}
		\otimes 
		\Big(	  \,  b_\ell^* b_k
		- 
		\, c_k^*  c_\ell  \Big) \  . 
	\end{align}
	Note $\V_- = ( \V_+)^*$. 
	Here, the notation is suggestive: Namely,
	$\V_+$ creates a collective excitation 
	of particle-hole pairs due to interactions with the Bose system.
	On the other hand $\V_-$ destroys such excitations. 
	The diagonal term $\V^{\rm diag}$
	does not create new particle-hole pairs, 
	and   in our analysis, it turns out as a lower order term. 
	We note the operators $  b_p, c_\ell $
	carry all of the notation regarding summation regions, see \eqref{def:bk:ck}.
	
	After application of  
	the particle-hole transformation, we can    relate  the original  eigenvalues
	$ \mu_n ( H )  
	$ with those of the excitation Hamiltonian $\mathbb H$. 
	To this end, 
	we first introduce the number operators 
	on $\mathscr H _{B } \otimes \F$
	\begin{equation}
		\mathcal  N_{+ }
		\equiv 
		\1 \otimes 
		\frac{1}{2}
		\sum_{ k \in \Z^3 }
		\big(		 b_k^*b_k  +  c_k^*c_k 			\big)
		\qquad 
		\t{and}
		\qquad 
		\mathcal  N_{ -  }
		\equiv 
		\1 \otimes 
		\frac{1}{2}
		\sum_{ k \in \Z^3 }
		\big(		 b_k^*b_k  -  c_k^*c_k 			\big) \ . 
	\end{equation}
	We refer to $\mathcal N _-$ as the \textit{charge operator}
	and its kernel 
	plays a major role in our analysis. It represents the subspace of states
	with equal number of particles and holes. 
	Throughout this article, we will write 
	\begin{equation}
		\K \equiv \ker (\calN_-) \subset \mathscr H_B \otimes \F \ . 
	\end{equation}
	We then observe that $   \K = \mathscr R^* \mathscr H$, that is, it is the image of the original Hilbert space under the particle-hole transformation. 
	This follows from 
	$
	\mathscr R^* \mathcal N    \mathscr R  
	= 
	M +  2 \mathcal N_- 
	$
	and the fact that $\mathcal N \upharpoonright \mathscr H   = M$. 
	Thus, we now have that 
	\begin{equation}
		\label{eigenvalues}
		\mu_n (H) 
		= 
		E_F + \lambda N M  \hat V_0  +
		\mu_n (  \mathbb H    \upharpoonright\mathscr K)  \  , 
	\end{equation}
	for all $ n \geq 1 $. 
	We observe that the restriction to $  \K $
	has the effect of making the  energy $\mathbb T$ non-negative.
	Namely, 
	on $    \K $
	we can add a term proportional to $\mathcal N_-$
	and write 
	\begin{align}\label{eq:rewriting:T}
		\mathbb T & =
		\sum_{ k \in \Z^3} 
		\big| k^2-	k_F^2 	\big|   \, b_k^* b_k + 
		\sum_{k \in \Z^3}
		\big|	k^2 - k_F^2\big|  \, c_k^* c_k  \  .
	\end{align}
	Unless confusion arises, 
	we ease the notation 
	and   write  
	$\mathbb H =  \mathbb H    \upharpoonright  \K $.

	As a concluding remark, 
	let us recall 
	that an advantage of 
	working with creation and annihilation operators
	is the use of \textit{pull-through formulas}.
	In this article, we will mostly use the following version, 
	in which we use the $b$ and $c$ operators, 
	and the   kinetic energy $\mathbb T$
	\begin{align}
		& 	f (	 \mathbb T	)  b_k^* = b_k^* f (\mathbb T  + k^2)	 \ , 
		\qquad 
		& &  b_k 	f (	 \mathbb T	)    =   f (\mathbb T  +  k^2)b_k  	  \ , \\
		& 	f (	 \mathbb T	)  c_k^* = c_k^* f (\mathbb T  - k^2)	  \  ,  
		& & 
		c_k   	f (	 \mathbb T	)   =   f (\mathbb T   - k^2)c_k   	  \  , 
	\end{align}
	for any $ k \in \Z^3 $, 
	and any measurable function $f : \R \rightarrow \R . $
	They follow directly from the CAR and the functional calculus for $\T$.

	\section{Proof of Theorem \ref{thm1} and \ref{cor}}
	\label{sec:proofs}
	
	\subsection{Proof of Theorem \ref{thm1} }
	First, let us  turn to the proof of   Theorem \ref{thm1}.
	As discussed in Section \ref{section:particle-hole}, 
	up to subtraction of constant terms, 
	the eigenvalues of $H$
	are determined 
	by  those of  a  Hamiltonian on Fock space, denoted  by $\mathbb H $.
	The operator $\mathbb H$
	describes  the interactions between the (slow) bosonic subsystem and the (fast) excitations of the Fermi ball.
	In particular, the asymptotics of its spectrum 
	can be analyzed in the framework of the theory described in Subsection \ref{subsec:eff}.
	This   argument is made rigorous in Section \ref{sec:effective}. 
	It turns out that the resulting effective theory
	is described by a
	$k_F$-dependent     Hamiltonian 
	$h_{\kf}^\eff$ on the bosonic space $\mathscr H_{\rm B}$.

	In order to state its definition, we must introduce the following 
	\textit{lune sets}
	\begin{equation}
		\label{lune}
		L(k) \equiv \{ p \in \Z^3 : |p| > k_F \,  \t{ and }  \, |p - k| \leq k_F \}  \ , \qquad k \in \Z^3 \  . 
	\end{equation}
	In particular, throughout this article we will consider various   sums
	suported on these sets. 
	These sums give rise to important coefficients, which enter the effective Hamiltonian.
	Here, the relevant one is given by 
	\begin{align}
		\label{eq:D}
		D( k , \kf)
		\equiv 
		\sum_{ p \in L(k)} 
		\frac{1}{p^2 - (p - k)^2}   \ , \qquad  k \in \Z^3 \ . 
	\end{align} 
	Let us comment that the same resolvent sums 
	appear in the analysis of the ground state energy of interacting Fermi gases  on the torus $\To^3$ in the mean-field limit; see \cite{HPR,Christiansen2023,
		Christiansen2023-1,Christiansen2023-2,
		Christiansen2024,Benedikter1,
		Benedikter2, Benedikter3}.
	In particular, we shall borrow certain estimates (and, prove new ones) which we state in Section \ref{sec:prelims}. 
	Now,  the effective Hamiltonian   takes the form 
	\begin{align}
		\label{heff}
		h_{k_F}^{ \rm eff }
		\equiv  \sum_{ i= 1 }^N( -\Delta_i) 
		+ 
		\frac{1}{N}
		\sum_{ 1 \leq  i < j \leq N }
		( W - W_{k_F})(x_i - x_j )
	\end{align}
	which is defined in terms of the effective potential defined on $\To^3$
	via the inverse Fourier formula
	\begin{equation}\label{eq:def:eff:kF:potential}
		W_{k_F}(  x) 
		\equiv   \frac{1}{2 \pi k_F}
		\sum_{k\in \Z^3} 
		e^{ i k  x }
		\, 	| \widehat V (k)|^2  	 \, D (k,k_F)   \ . 
	\end{equation} 
	
	In a second step, we compare the 
	eigenvalues of  $h_{\kf}^\eff$ with those of   $h^\eff$ introduced in \eqref{h:eff}.
	For this we use the   asymptotics for the coefficients  as $\kf \rightarrow \infty$
	\begin{equation}
		D( k ,\kf) = 2 \pi \kf + o (\kf) \ , \qquad \forall k\in \Z^3 /\{	0	\}  \ . 
	\end{equation}
	In fact, we use a more precise version stated in Section \ref{sec:prelims}. 
	Note that, combined with \eqref{eq:def:eff:kF:potential}, 
	this will give rise to the convolution structure $V*V$ of the effective potential.
	It is, however, necessary to isolate the zero-Fourier modes
	which give rise to additional constant terms in the energy expansion for $\mu_n(H)$.

	%
	%We remind the reader   that the coefficients
	%$\D_1(k,k_F)$ were defined in  \eqref{eq:def:S:alpha:k}. 
	%In particular, 
	%from Lemma  \ref{lemmaS1}
	%%we see that $\| W_{\kf}\|_{L^\infty} \leq  C \| V  \|_{H^2}^2$
	%%for some constant $C>0$. 
	%%For the second step, 
	%%we prove that 
	%%$ \lim_{k_F \rightarrow \infty } W_{k_F}  = V *V $ with explicit error rate, and use this to relate the eigenvalues of $h_{k_F}^{ \rm eff }$ with those of $h^{ \rm eff }$.
	%%Here, the main ingredient 
	%%is the asymptotic of the coefficients
	%\begin{equation}
	%	\label{asymp}
	%	|	 \mathcal D _1 (k , k_F) - 2 \pi  k_F 	|  
	%	\leq C (\ln \kf)^{5/3} \kf^{2 /3 }  \ , \qquad \forall k\in\Z^3  / \{ 0 \} \ . 
	%\end{equation}
	%
	%
	\medskip 
	Before we turn to the proof
	let us introduce the following 
	quantity 
	\begin{equation}
		\label{eq:q}
		\Q_p  
		\equiv  
		1+ 
		\|  W\|_{L^p }^{2 p/ (2p-3) 	}
		+  
		\| 	 V 	 \|_{H^4 }^2    \ . 
	\end{equation}
	We also recall  $\K = \ker (\calN_-)$. 
	
	\begin{proof}[Proof of Theorem \ref{thm1}]
		Throughout this proof we fix $k_F \geq 1 $ and $N \geq 1 $.
		\medskip 
		
		\noindent  {\underline{First step}}. Recall that the effect of the  conjugation with the particle-hole transformation
		\eqref{eigenvalues}  is that we have 
		\begin{equation}
			\label{eq:eigen1}
			\mu_n( H )
			= 	E_F  +\lambda NM 
			\hat V_0 
			+ 
			\mu_n ( \H ) \  , \qquad n \geq 1  \ ,
		\end{equation}
		where $\mathbb H  = \mathbb H \upharpoonright\mathscr K $ denotes the operator \eqref{def:hamiltonian:exc}
		restricted to the subspace  $\K  $. 
		
		In Proposition~\ref{prop2}, we provide an upper bound for the eigenvalues of $\mathbb {H}$, while a corresponding lower bound is established in Proposition~\ref{prop1}.
		
		These estimates  imply that we can find   a constant $C>0$  
		such  that 
		for all $1 \leq n \leq  N$ 
		\begin{equation}
			\label{eq4.2}
			\Big| 	 \mu_n ( \H ) - \mu_n (	  h^{\rm eff}_{ k_F } ) 
			+  \tfrac12    W_{\kf}( 0)
			\Big| 
			\leq     C \Q_p^2  N^2   \,  (\ln \kf)^{5/3} \kf^{-1/3 } 
		\end{equation}
		provided
		$k_F  (\ln k_F)^{-5}  \geq C (n  N  \| V\|_{H^4}^2  )^3$. 
		Here $h_{\kf}^\eff$ is the effective Hamiltonian introduced in  \eqref{heff}.
		Thus, thanks to \eqref{eq4.2} it is sufficient  to analyze the difference between 
		$h_{\kf}^\eff$ and $h^\eff$, as well as the leading order term of $\tfrac12 W_{\kf}(0).$

		\medskip 
		\noindent 	\underline{Second step}.
		Let us now compare the eigenvalues of  $h_{\kf}^\eff$ and $h^\eff $.
		To this end, 
		we will consider the Fourier representation for $x \in \To^3$
		\begin{equation}
			(  V *V)  (x)
			= 
			\sum_{ k \in \Z^3 } 
			e^{ik   x } \, 
			| \widehat V(k)|^2     \ . 
		\end{equation}
		Observe that $L(0)=\varnothing$ and thus $D( 0 , \kf)=0$. 
		Hence, we can write  for $x \in \To^3$
		\begin{align}
			\notag
			W_{k_F} (x) 
			&    = 
			\sum_{  k \neq 0 }
			e^{i k x}
			|\widehat V(k)|^2 
			\frac{	    D  (k , k_F) }{2\pi k_F} \\
			\notag
			& 	   = 
			\sum_{  k \neq 0 }
			e^{i k x}
			|\widehat V(k)|^2 
			+
			\sum_{  k \neq 0 }
			e^{i k x}
			|\widehat V(k)|^2 
			\Big(
			\frac{	   D  (k , k_F) }{2\pi k_F} 
			-1 
			\Big) \\
			& = 
			(   V*V)  (x)
			- | \widehat V (0)|^2 
			+ 
			\sum_{  k \neq 0 }
			e^{i k x}
			|\widehat V(k)|^2 
			\Big(
			\frac{	    D (k , k_F) }{2\pi k_F} 
			-1 
			\Big)
			\label{eq:w}
		\end{align}
		The very last sum  sum over $ k \in \Z^3 / \{ 0 \}$
		constitutes the main error term.
		The summation may be split into two regions 
		$0 < | k | < 2 k_F$  
		and
		$ | k | \geq 2  k_F$.
		For the first region, we use \eqref{eq:sum} from Lemma \ref{prop:sums}, whereas
		for the second region we can use 
		\eqref{eq:S1:k:large} from Lemma \ref{prop:sums}.
		We obtain the inequalities 
		\begin{align}
			\bigg|
			\sum_{  0 < | k |< 2 \kf  }
			e^{i k x}
			|\widehat V(k)|^2 
			\Big(
			\frac{	    D (k , k_F) }{2\pi k_F} 
			-1 
			\Big)
			\bigg|
			&  \ 	\leq  \ 
			C 	
			\sum_{  0 < | k |< 2 \kf  } 
			|\widehat V(k)|^2 
			|k|^{ 4  } (  \ln  k_F)^{ 5/3 } k_F^{	 -1/3 }  \\
			\bigg|
			\sum_{   | k | \geq  2 \kf  }
			e^{i k x}
			|\widehat V(k)|^2 
			\Big(
			\frac{	    D (k , k_F) }{2\pi k_F} 
			-1 
			\Big)
			\bigg|
			&  \    \leq  \ 
			C 	
			\sum_{   | k | \geq  2 \kf  }
			|\widehat V(k)|^2  
			\     \leq  \ 
			\frac{C}{\kf^4 }
			\sum_{   k\in \Z^3  }
			|\widehat V(k)|^2   |k|^4  \ . 
		\end{align}
		These  two estimates combined with \eqref{eq:w}
		imply the  $L^\infty$ estimate
		\begin{align}\label{supremum:difference}
			\|	 W_{k_F} - (V*V) + | \widehat V_0	|^2	\|_{L^\infty (\To^d )} 
			\leq C (\ln \kf)^{5/3} \kf^{-1/3 } \|  V  \|_{H^2}^2 \ . 
		\end{align}
		%	We then readily obtain that
		%	\begin{align}\label{supremum:difference}
			%		\sup_{x\in \To^3} \Big|   (  V\ast V ) (x)    - W_{k_F}(x) \Big|  \leq 
			%		\sum_{ k \in \Z^3 }
			%		| \widehat V(k)	|^2 
			%		\, 	\Big| 
			%		1 - \frac{ \D_1 (k ,k_F)}{ 2 \pi k_F }
			%		\Big| 
			%		\equiv 
			%		I_1 + I_2  
			%	\end{align}
		%	where $I_1$ is the sum over 
		%	$| k | < 2 k_F$  
		%	and $I_2$ over
		%	$ | k | \geq 2  k_F$.
		%	We use the estimates for $\D_1$ in Eqs. 
		%	\eqref{eq:sum} and \eqref{eq:S1:k:large}, respectively, 
		%	to obtain that for some constant $C>0$  
		%	\begin{align}
			%		\label{eq2}
			%		I_1 
			%		\  	\leq \ 
			%		\frac{C (\ln k_F)^{ 5/ 3  }}{ k_F^{1/3 }} 
			%		\|   V \|_{H^2}^2   
			%		\Q^q quad  
			%		\t{and}
			%		\Q^q quad 
			%		I_2   
			%		\ \leq \ 
			%		\frac{C}{k_F^{2/3}} 
			%		\|  V \|_{H^1 }^2 \ . 
			%	\end{align}
		As a   consequence,
		summing over all particle possibilities $x_i  -x_j$, 
		we have in the sense of quadratic forms 
		on $\mathscr H _{B}$ that
		\begin{align}
			\label{eq1}
			\pm \big(
			h^{ \eff}_{ \kf} 	-	  		h^{\rm eff}    
			+   \tfrac12 ( N -1 ) 	 | \widehat V_0|^2
			\big) 
			\leq C N  (\ln \kf)^{5/3} \kf^{-1/3 } \|  V  \|_{H^2}^2  
		\end{align}
		and thus, via  the min-max principle \cite[Theorem 4.14]{Teschl2014}, we obtain 
		\begin{equation}
			\label{eq4.1}
			\Big|	
			\mu_n (	  h^{\rm eff}_{ k_F} 	)		 - 
			\mu_n (  h^{\rm eff} )	   
			+   \tfrac12 ( N -1 )  | \widehat V_0	|^2	\Big|   
			\leq  C N  (\ln \kf)^{5/3} \kf^{-1/3 } \|  V  \|_{H^2}^2  
			\qquad \forall n \geq 1 \ . 
		\end{equation} 
		This controls the difference between the eigenvalues of $h^\eff$ and $h_{\kf}^\eff$. 
		
		Finally	the bound  \eqref{supremum:difference}
		at $x = 0$ also implies 
		\begin{equation}
			W_{\kf}(0) = 
			\int_{\To^3} |V(x)|^2 \d x   - | \widehat V_0|^2 + O (	  (\ln \kf)^{5/3} \kf^{-1/3} \|  V  \|_{H^2}^2    ) \ . 
		\end{equation}
		Thus, we combine  \eqref{eq:eigen1},    \eqref{eq4.2}  and \eqref{eq4.1}  to finish the proof. 
	\end{proof}
	
	The remainder of the article is essentially devoted to establishing the upper and lower bounds for $\mu_n(\mathbb{H})$ given in Propositions~\ref{prop2} and~\ref{prop1}, respectively. The proofs of these results rely on a number of technical lemmas, including suitable sum estimates and operator bounds. These auxiliary results are stated in Section~\ref{sec:prelims}. Before proceeding with this program, however, we provide the proof of Theorem~\ref{cor}.
	
	\subsection{Proof of Theorem \ref{cor}} 
	
	We note that by Theorem \ref{thm1} and \eqref{eq:H:n:convergence}, the energy-per-boson is given by 
	\begin{align}\label{eq:proof:formula:E}
		\mathcal E(g) = 	\liminf_{N\to \infty}    \frac1N \inf \sigma(h^{\rm eff}|_g),
	\end{align}
	and thus the analysis can be reduced
	to studying the spectral properties 
	of the bosonic Hamiltonian $h^\eff$ and its dependence on $g \in \R$; recall its definition in \eqref{eq:heff}. 
	Let us immediately note that $\mathcal E(g) \leq \mathcal E(0) < \infty$ thanks to \eqref{eq:stable:h:n:energies}.

	For the proof, we show that whenever $w_g = w - g^2 (   v*v) $
	is negative at $x=0$, the energy-per-particle collapses to $ - \infty$ as $N\to \infty$. 
	Using a monotonicity argument, we prove the existence of the \textit{smallest} $g_c \ge 0$ for which this collapse occurs.

	\begin{proof}[Proof of Theorem \ref{cor}] 
		Thanks to \eqref{eq:proof:formula:E}, it
		suffices to look at $\inf \sigma (h^\eff|_g)$
		and its dependence on $g$. 
		
		(1) Since $h^{\eff}|_g$ depends only on $g^2$, we immediately obtain $\mathcal E(g) = \mathcal E(-g)$.
		Next, observe that the map 
		$ [ 0 , \infty )  \ni g \mapsto  \mathcal E (g)  = \liminf \tfrac1N \inf \sigma ( h^\eff |_g )$
		is monotone decreasing.
		To see this,  let $ g_2 > g_1>0$
		and use
		$v*v \geq 0 $:
		\begin{align}
			(g_2^2-g_1^2) \langle \Psi , \textstyle \sum_{1\le i < j \le N} v\ast v(N(x_i-x_j)) \Psi \rangle \ge 0
		\end{align}			
		for all
		$  \Psi \in \mathscr H_B  $. 
		Thus,   $ h^{\eff} |_{ g = g_1} \geq  h^{\eff} |_{ g = g_2}$
		and we conclude that 
		$ \mathcal E (g_1) \ge  \mathcal E (g_2)$ for $ g_2 \ge g_1 \ge 0 $.
		
		(2) Now, for the proof of the main statement within Theorem \ref{cor},  	we set
		\begin{equation}
			g_c \equiv \inf \big\{ g \geq 0 :  
			\liminf_{N \rightarrow \infty} N^{-1 } \inf \sigma (h^\eff  |_{g  }) = - \infty		\big\}  \in [0 , \infty] \  . 
		\end{equation} 
		Let us first show that     $g_c < \infty$.
		Indeed, since $ (v * v ) (0) = \|v \|_{L^2}^2>0$
		we have $w_g(0)< 0 $ for all $g >  g_* \equiv  \frac{ w(0)}{  (v*v)(0) } $. We now claim that
		\begin{equation}
			\label{w*}
			\lim_{N\to \infty}   N^{-1} \inf \sigma(h^\eff|_g ) = - \infty  
			\qquad \t{ for }
			g > g_* \ .
		\end{equation}
		Fix $g > g_*$, and note that by continuity of $w_{g}$,
		there exist $c_0,r_0>0$ so that 
		$
		w_{g}(x) \leq - c_0 < 0   
		$
		on the ball $ B_0 \equiv B(0, r_0)$. 
		We now construct an appropriate trial state. 
		Take a radial symmetric function $ \psi  \in C_0^2 (\R^3)$    so  that 
		$ \supp (\psi) 
		\subset B  ( 0, \frac{r_0}{2} )
		$ and $\|  \psi\|_{L^2}=1 $. 
		In particular, $
		\supp (
		|\psi|^2 *  |\psi|^2 
		) \, \subset  \,
		B_0 $.
		We define 
		$
		\psi_N(x)   = N^{3/2} \psi ( N x )    
		$ 
		and set  $\Psi_N = \psi_N^{\otimes  N }$ so that
		\begin{equation}
			\< \Psi_N   ,  \textstyle   \sum_{ 1 \leq i \leq N } (- \Delta_i) \Psi_N \> 
			\ 	  =   \ 
			N^3 \|  \nabla \psi  \|_{L^2}^2	\, .
		\end{equation}
		On the other hand, using radial symmetry, we find for the potential energy 
		\begin{align} 
			\< \Psi_N  , \textstyle \frac{1}{N} \sum_{1 \leq i < j \leq N } N^3 
			w_g (N (x_i - x_j))					\Psi_N \> 
			=   N^4 \textstyle
			\int_{B_0 } ( |\psi  |^2 \ast |\psi 	|^2)  w_g \, dx  
			\leq 
			-  c_0  \, N^4 \ . 
		\end{align}
		Thus
		$
		\<\Psi_N , h^{\rm eff}  \Psi_N\>
		\leq  - N^4  (	  c_0  -  N^{-1} \|  \nabla \psi  \|^2       ) 	  $ 
		from which we conclude  \eqref{w*} for $g>g_*$.
		Therefore, $g_c  \leq  g_* < \infty$.
		
		\medskip 
		
		(3) From the previous argument we already  conclude 
		$g_* \geq g_c$. 
		On the other hand, for  $  g   \in \{ 0 \} \cup ( 0 ,g_0) $ with $g_0 = \inf \{ g \ge 0 : w_g \ge 0 \} \ge 0$, the statement $\mathcal E(g) = 4 \pi \mathfrak a(g)$ is a direct consequence of \eqref{eq:proof:formula:E} and \eqref{eq:stable:h:n:energies}. 	
	\end{proof}

	\section{Auxiliary estimates}
	\label{sec:prelims}
	In this section we prove 	 some necessary estimates
	for the proof of the upper and lower bounds. 
	First, we state estimates for sums involving the free resolvent for particle-hole excitations.
	Second, we prove some relevant operator inequalities.

	\subsection{Sum estimates}\label{sec:sum}
	Here and in the sequel, 
	we will consider various asymptotics of sums which are supported on the lune sets $L(k)$,  
	defined in \eqref{lune}. 
	The coefficients $D(k, \kf)$ introduced in  \eqref{eq:D}
	are the main example.
	However,  for our analysis it will be necessary  to   generalize  these sums. 
	Namely, 
	for any $\alpha \geq 0 $ we let 
	\begin{align}\label{eq:def:S:alpha:k}
		\mathcal  D_\alpha(k,k_F) 
		\equiv
		\sum_{  p \in L(k)   } \frac{1}{( p^2 - ( p - k)^2  )^{\alpha}}   \ , \qquad k \in \Z^3 \ .
	\end{align}
	Let us note that, by definition, $L(0) $ is the empty set
	and thus   $\mathcal D_\alpha ( 0 , k_F)=0$.
	
	\medskip
	
	Before we turn to these estimates, we
	observe the following useful fact: 
	for all  non-zero $ k \in \Z^3$, we have
	$p^2  - (p - k)^2 \geq 1 $ for all $ p \in L(k)$. 
	In particular, thanks to the crude geometric estimate $|L(k)| \leq C \kf^{3 } $
	one has the bound $\D_\alpha (k ,\kf) \leq C \kf^3$. 
	
	\medskip
	
	The following lemma contains 
	the main estimates that we use to analyze the coefficients $\D_\alpha$. 
	The inequalities for $\D_1 =D $ will be borrowed from \cite{Christiansen2023}.
	The ones for $\D_2$ are---to our knowledge---new, 
	but can be obtained from the methods   developed in  \cite{Christiansen2023}.

	\begin{lemma}
		\label{prop:sums}
		There exists a constant $C>0$ such that the following holds
		for all $\kf \geq 1 $. 
		
		\begin{enumerate}
			\item  For all  $ 0< | k | < 2  k_F $ 
			\begin{equation}
				\label{eq:sum}
				\Big|     \D_1 (k, \kf)
				- 2\pi  \kf 
				\Big|   
				\leq C  |k|^{ 4  } (  \ln  \kf )^{ 5/3 } \kf^{	 2/3 } 
			\end{equation}

			\item 
			For all  
			$ 0 < |k | < 2  k_F$ 
			\begin{align}
				\label{eq:bound:S1}
				\D_1(k,k_F)
				&   \ 	\leq  \ 
				C  \, k_F  \, \big(1 +	 (  \ln  k_F)^\frac{5}{3}  k_F^{-	\frac{1}{3}}		\,  |k|^{ 4 }  	 \big)   \\
				\D_2(k,k_F)
				& 	 \ \leq  \ 
				C  \, k_F^\frac{2}{3} \,    (  \ln  k_F)^\frac{2}{3}  \, |k|^{ 4 }   	\label{eq:bound:S2}  \ . 
			\end{align}
			
			\item 
			For all  
			$| k | \geq 2 k_F $ 
			\begin{equation}\label{eq:S1:k:large}
				\D _1(k,k_F) \ \leq \  C\,  k_F^3  |k|^{-2 } \ . 
			\end{equation}
			
		\end{enumerate}	
	\end{lemma}
	
	\begin{proof} 
		The   inequality in (1) is taken from  \cite[Proposition B.21]{Christiansen2023}.  
		The first inequality in (2) follows directly from the previous one.
		The proof of the  second inequality in (2) is more involved and postponed to Section \ref{section:sums}. 
		The   inequality in (3)  is taken from \cite[Proposition B.9]{Christiansen2023}. 
	\end{proof} 
	
	%\label{eq:bound:S2:b}  \ . 
	
	%
	%\begin{remark} To simplify the upcoming notation, we shall proceed using the bounds:
	%	\begin{align}
		%		\D_1(k,k_F)
		%		&   \ 	\leq  \ 
		%		C  \, k_F  \, \big(1 +	 (  \ln  k_F)^\frac{5}{3}  k_F^{-	\frac{1}{3}}		\,  |k|^4  	 \big)  , \quad 
		%		\D_2(k,k_F)  	 \ \leq  \ 
		%		C  \, k_F^\frac{2}{3} \,    (  \ln  k_F)^\frac{5}{3}  \, |k|^4  	\label{eq:bound:S2:b}  \ . 
		%	\end{align}
	%\end{remark}

	Our next goal is to combine the previous sum estimates
	with decay properties of the Fourier coefficients $\widehat  V (k)$.
	To this end, let us  introduce the following notation 
	\begin{align}
		\label{def:S}
		\mathcal S_{\alpha,\beta}(k_F) &  \equiv  
		\textstyle 
		\sum_{ k\in \Z^3 } |\widehat V(k)|^2
		\<  k  \>^{2 \beta }
		\D_\alpha(k, k_F) 
	\end{align}
	for $ k_F \geq 1 $ and $\alpha ,\beta \geq 0$.
	Here,   $\<k\>   \equiv    (1 + |k|^2)^{1/2}$. 
	
	\begin{lemma}\label{lemmaS1} 
		There is   $C>0$  such that for all $k_F \geq  1 $ and $\beta \ge 0$
		\begin{align}
			\label{S1}
			\mathcal S_{1,\beta}(k_F) &
			\, \leq  \,  C	  \,	 \|	 V  \|_{H^{2+\beta }}^2  \, k_F   \\
			\label{S2}
			\mathcal S_{2,\beta}(k_F) & 
			\, \leq   \, 
			C    \, 
			\|	 V  	\|_{H^{2+\beta}}^2
			\, 	k_F^{2/3}			(	 \ln k_F)^{5/3} \  .
		\end{align}
	\end{lemma}
	
	\begin{proof} 
		Fix $\beta \geq 0$ and $\kf\geq 1 $, and   write  for short $\mathcal S_{\alpha,\beta}   = \mathcal S_{\alpha, \beta} (\kf )$, 
		$\D_\alpha (k)  =  \D_\alpha (k,\kf)$. 
		Denote $f (k)  = | \widehat V(k)|^2 \< k\>^{2\beta}	$. 
		We  split the region of summation as 
		\begin{equation}
			\notag
			\mathcal S _{1} 
			\ = \ 
			\sum_{ | k |  <  2 k_F	}	
			f(k)   \D _1(k ) 
			+ 	  
			\sum_{ | k |  \geq  2 k_F	}
			f(k) 
			\D_1(k )  
		\end{equation}
		For the first region, we find with \eqref{eq:bound:S1}
		\begin{align*}
			&   	 \sum_{ | k |  <  2 k_F	}    
			f(k) 
			\D _1(k )  \\
			&  \ \leq  \  	  	 		    	   	  	 		      
			C \kf    \sum_{ | k |  <  2 k_F	}f (k) 
			\Big (1 +
			\frac{	  ( \ln \kf)^{5/3}  }{ \kf^{1/3 }}	|k|^4 		\Big)	 
			\ 	 \leq	\	  C \kf 
			\sum_{ k \in \Z^3} f(k) (1 + |k|^4)
			\ \leq  \  C \kf \| V\|_{H^{\beta + 2}}
			\ .
		\end{align*}
		For the second region, we find thanks to  \eqref{eq:S1:k:large}
		\begin{align*}  
			\sum_{ | k |  \geq   2 k_F	}
			f(k)  \D _1(k) 
			\ \leq \ 
			C 	\sum_{ | k |  \geq   2 k_F	}	
			f(k)    \kf^3 |k|^{-2 }
			\		  \leq	\	  C  \kf 
			\sum_{ k \in \Z^3} f(k)   
			\    = \  C \kf \| V \|_{H^{2 \beta }}\ .
		\end{align*}
		This finishes the proof for $\mathcal S_1$. 
		\smallskip

		For the estimates regarding  $\mathcal S_2  $ 
		we split the sum into the same two regions. 
		For the small region $ | k |< 2 \kf$	 we use  the bound  \eqref{eq:bound:S2} from Lemma \ref{prop:sums}.
		For the large region $|k| \geq 2 k_F$. 
		we use the crude bound  $\mathcal \D_2(k) \leq  C k_F^{3} $ instead. 
		Thus, we find that for all  $\gamma>0$
		there is $C>0$
		such that 
		\begin{align*}
			\mathcal S_2
			& 	    \ = \ 
			\sum_{ | k |  <  2 k_F	}	
			%		 |\widehat V (k)|^2
			%		\<  k  \>^{2 \beta } 
			f(k) 
			\D _2(k ) 
			\  + \ 
			\sum_{ | k |  \geq  2 k_F	}	
			%		 |\widehat V (k)|^2     	  \<  k  \>^{2 \beta }    
			f(k) 
			\D _2(k )   \notag   \\ 
			&  \ \leq \ 
			C k_F^{2/3} ( \ln  k_F)^{5/3}  
			\sum_{ | k |  <  2 k_F	}	
			%		 | \widehat V (k)|^2    	  \<  k  \>^{2 \beta } 
			f(k)    |k|^4 
			\ + \ 
			C k_F^{3-\gamma }	 
			\sum_{ | k |  \geq  2 k_F	}	
			%		 | \widehat V (k)|^2    	  \<  k  \>^{2 \beta }  
			f(k) 
			|k|^\gamma \notag \\
			&      \ \leq \ 
			C k_F^{2/3} ( \ln  k_F)^{5/3} \| V \|_{H^{\beta +2}} 
			\ + \ 
			C k_F^{3-\gamma}	  
			\| V \|_{H^{\beta +\gamma/2}} \ .
		\end{align*}
		Setting $\gamma = 7/3$ completes the proof of the second bound. 
	\end{proof}
	
	\begin{remark}
		As a useful corollary of Lemma \ref{prop:sums} with $(\alpha,\beta) = (1, 0)$
		we 
		obtain the following $L^\infty$ estimate
		for the effective potential $W_{\kf}$ given by \eqref{eq:def:eff:kF:potential}
		\begin{equation}
			\label{eq:W:bound}
			\|  W_{\kf} \|_{L^\infty (\To^3)} 
			\leq   C \| V\|_{H^2}^2   \ . 
		\end{equation}
	\end{remark}

	\subsection{Operator estimates}
	
	First, we prove that the number of excitations above the Fermi sea can be controlled by their kinetic energy. 
	Recall $\K = \ker (\calN_-) \subset \mathscr H_B \otimes \F $ . 
	
	\begin{lemma}
		%	[Kinetic lower bound]
		\label{lemma:kinetic:term}
		For all  $\Psi \in 
		\K$, the following inequality holds
		\begin{equation}
			\< \Psi, \mathcal N_+ \Psi  \> \  \le \  	
			\<  \Psi,  \mathbb T  \Psi\> \ .
		\end{equation} 	 
	\end{lemma}
	
	\begin{proof} 
		Let $ n \in \Z $ be the greatest integer 
		such that $ n \leq k_F^2$.
		Then, for all $ p \in B^c$ and $h \in B$	  it holds that: 
		$ 	p^2 \geq n + 1 $   and $  h^2 \leq n $. 
		Thus, from \eqref{eq:T} we have the lower bound
		\begin{align}
			\notag
			\<  \Psi, \mathbb T \Psi \>
			& 	 
			\ =  \ 
			\sum_{ p\in B^c}
			p^2 \|	 b_p \Psi		\|^2 
			- 
			\sum_{  h \in B }
			h^2 \|	 c_h \Psi		\|^2   \\
			\notag
			& 	\    \geq  \ 
			( n + 1 )	 \sum_{ p\in B^c}
			\|	 b_p \Psi		\|^2 
			-     n 
			\sum_{ h\in B}
			\|	 c_h \Psi		\|^2     \\
			& \ = \ 
			\sum_{ p\in B^c}
			\|	 b_p \Psi		\|^2  
			+ n \Big(
			\sum_{ p\in B^c}
			\|	 b_p \Psi		\|^2 
			-     
			\sum_{ h\in B}
			\|	 c_h \Psi		\|^2  
			\Big)  \ . 
		\end{align}
		Observe now that the term with $ n $ 
		is proportional to $\< \Psi, \calN_- \Psi \> = 0.$
		On the other hand 
		$   \langle \Psi ,  \sum_{p\in B^c} b_p^* b_p  \Psi \rangle = 
		\langle  	\Psi , ( \mathcal N_+ 	+ \mathcal N_-	) \Psi \rangle $. This finishes the proof. 
	\end{proof}

	Next, we will need an estimate for the diagonal operator $\V^{\t{diag}}$.
	In most cases, it can be safely regarded as a lower order term. 
	\begin{lemma}
		%	[Diagonal terms]
		\label{lemma:diagonal:terms} 
		For all $\Psi \in \mathscr H _{B} \otimes \mathscr F     $
		\begin{equation}
			|  	\<  	\Psi, 	 \V^{\rm diag} 	 \Psi		\> |  \leq 2 N  \|  \widehat V\|_{\ell^1 }     \<	 \Psi, \mathcal N_+ \Psi	\> \ . 
		\end{equation}
	\end{lemma}
	
	\begin{proof}The statement follows from
		\begin{align}
			| 	\<	 \Psi,  \V^{\rm diag}   \Psi\> | 
			&  \, \leq  \,  \tfrac12 N
			\sum_{\ell,k\in \Z^3 } | \widehat V(\ell-k)| \Big( \| b_\ell \Psi \|^2 + \| b_k \Psi\|^2 + \| c_\ell \Psi \|^2 + \| c_k \Psi\|^2 \Big) \notag\\
			& \,  \le	\,	 2 N  \| \widehat V\|_{\ell^1 }   \|	 \mathcal N_+^{1/2} \Psi	\|^2 \ .
		\end{align}
		Here, 	  we used that $\sup_{k} \sum_{\ell} |\widehat V(\ell-k) | = \| \widehat V \|_{\ell^1} $ in the last line. 
	\end{proof}
	
	The next lemma is a small modification of \cite[Lemma 4.7]{HPR}. It will be useful to bound the off-diagonal terms $\V_-$.   
	
	\begin{lemma} \label{lemma:HPR} 
		There is a constant $C>0$ such  that for all  $\alpha  \geq 0 $, $\varepsilon >0$ and 
		$ 	\Psi\in   \K  $
		\begin{align*}
			\lambda  \sum_{k\in  \Z^3 } 
			| \widehat V(k)|
			\sum_{p\in L(k) }
			\frac{ \| c_{p - k} b_{p} \Psi \| }{   (  p^2  - (p - k)^2 )^{	\frac{\alpha-1}{2}		}}
			\, \leq \, 
			\ve   
			\<  \Psi,  \mathbb T  \Psi\>  
			\,  + \, 
			\frac{ C \,  \lambda^2 }{\ve  }
			\mathcal S_{\alpha,2} (k_F)
			.
		\end{align*}
	\end{lemma}
	\begin{proof}
		Let 
		$  		 F(k)    \equiv 
		\sum_{p\in L(k) }
		(  p^2  - (p - k)^2 ) 
		\| c_{p  - k} b_{ p} \Psi \|^2$
		and recall
		$\D_\alpha(k) \equiv \D_{\alpha} ( k ,\kf)$ was defined in \eqref{eq:def:S:alpha:k}. 
		First, the Cauchy-Schwarz inequality over $ p \in L(k)$ shows that
		\begin{equation}
			\sum_{p\in L(k) }
			\frac{ \| c_{p - k} b_{p} \Psi \| }{   (  p^2  - (p - k)^2 )^{	\frac{\alpha-1}{2}		}}
			\leq 
			F(k)^{1/2} \D_{\alpha}( k )^{1/2} \ . 
		\end{equation}
		Multiplying with $\lambda |\widehat V(k)|$,  summing over $ k \in \Z^3$, 
		and applying  the Cauchy-Schwarz again, 
		we see    
		that for all $\ve ' >0 $
		\begin{align*}
			\lambda	\notag  
			\sum_{k\in  \Z^3 } 
			| \widehat V(k)|
			\sum_{p\in L(k) }
			\frac{ \| c_{p - k} b_{p} \Psi \| }{   (  p^2  - (p - k)^2 )^{	\frac{\alpha-1}{2}		}}
			& 	
			\  \leq     \ 
			\lambda	\sum_{k\in  \Z^3 } 
			| \widehat V(k)| 
			F(k)^\frac{1}{2}
			\D_{\alpha}( k)^{\frac{1}{2}} 
			\\ 
			&		\  \leq  \ 
			\lambda	\bigg( 
			\sum_{ k \in \Z^3 }  
			\<  k  \>^{  - 4  } 
			F(k)
			\bigg)^\frac{1}{2}
			\bigg( 
			\sum_{ k \in \Z^3 } 
			| \widehat V(k)|^2
			\<  k  \>^{ 4  }  
			\D_{\alpha}( k) 
			\bigg)^\frac{1}{2}
			\\
			&   \ \leq \ 
			\frac{\ve ' }{2  }
			\sum_{ k\in \Z^3 } 
			\<  k  \>^{  - 4  }  F(k)   
			\    	  +  \ 
			\frac{  \lambda^2 }{2 \ve  '  }
			\,   \mathcal S_{\alpha,2} (k_F) \ . 
		\end{align*}
		We now estimate $\sup_{ k \in \Z^3} F(k)$ in terms of the kinetic energy.
		Indeed, we obtain, thanks to the operator  norm bounds
		$\|  b \| = \| c \| =1 $
		\begin{align}
			\notag
			F(k)  
			&  \ =  \ 
			\sum_{p\in L(k) }\Big(  p^2 - k_F^2 \Big )
			\| c_{p - k} b_{p } \Psi \|^2 
			+ 
			\sum_{p\in L(k) } \Big(  k_F^2 - (p - k)^2  \Big) 
			\| c_{p - k} b_{p } \Psi \|^2  \\
			\notag
			&  \  \leq  \ 
			\sum_{p\in B^c  }\Big(  p^2 - k_F^2 \Big )
			\|  b_{p } \Psi \|^2 
			+ 
			\sum_{  h \in B  } \Big(  k_F^2 - (p - k)^2  \Big) 
			\| c_h  \Psi \|^2  \\
			&  \ \leq   \ 
			\| 	 \T^{1/2} \Psi		 \|^2 
		\end{align}
		thanks 
		to the representation 
		\eqref{eq:rewriting:T} in the zero charge subspace.
		This finishes the proof by setting 
		$\ve '  =  2 \ve /C $
		for $C = \sum_{ k\in \Z^3 } 
		\<  k  \>^{  - 4  }  $. 
	\end{proof}

	Finally, we establish the following bound on the kinetic energy of the particle-hole excitations. By Lemma \ref{lemma:kinetic:term}, this immediately implies a bound for $\langle \Psi,\mathcal N_+ \Psi \rangle$ as well.
	
	\begin{lemma}
		%	[Kinetic energy in low-energy states]
		\label{lemma:kinetic1}
		There is a constant $C>0$
		such that for  all $  \delta >0 $
		and 
		all        
		$\Psi \in  \K     $
		\begin{align}
			\langle \Psi , \mathbb T \Psi  \rangle 
			\, \leq \, 
			(1 + \delta  ) \, 
			\langle \Psi , \mathbb H \Psi  \rangle  + C  N 
			\(      \delta^{-1 }  \|   V \|_{H^4 }^2 +   \| W \|_{L^p}^{\frac{2p}{2p-3 }}   \)  
			\| \Psi\|^2 \ . 
		\end{align}
		provided  $k_F \geq  C  \delta^{-2} N \| \widehat V\|_{\ell^1}^2   $.
	\end{lemma}
	\begin{proof} 
		Let $\delta>0$.	Without loss of generality, we assume thrugh the proof that $\| \Psi\|=1$.
		From the definition of $\H$ we find 
		\begin{equation}
			\< \Psi, \T \Psi \>
			= 
			\< \Psi, \H \Psi \>
			- 
			\lambda 
			\< \Psi, \V^{\rm diag} \Psi \>
			-2 \lambda \re 	 	 \< \Psi, \V^- \Psi \>\,  
			- \< \Psi, h \otimes \1 \Psi \> \ . 
		\end{equation}
		We analyze each term separately. First, for the diagonal part, 
		we use   Lemmas \ref{lemma:diagonal:terms} and \ref{lemma:kinetic:term} to bound 
		\begin{align}\label{eq:diag:bound}
			\lambda \, | 		 	  
			\< \Psi, \V^{\rm diag} \Psi \>
			|  
			\ 	 \leq \ 
			2 N^{1/2}  \kf^{-1/2 }
			\|  \widehat V\|_{\ell^1 } 
			\< \Psi, \T  \Psi \> 
			\ 		\leq  \ 
			\delta 		\< \Psi, \T  \Psi \>  
		\end{align}
		where the last inequality holds for  $k_F \geq 4 N \| \widehat V\|_{\ell^1}^2 \delta^{-2} $. 
		Secondly, 	for the off-diagonal part, 
		starting from the definition of $\V^-$ in \eqref{V-}
		we use the triangle inequality,  
		Lemma \ref{lemma:HPR} for $\alpha =1$, $\beta =2$ and $\ve = \frac{\delta}{N} $, 
		and 
		then \eqref{S1} from Lemma \ref{lemmaS1} to find 
		\begin{align}
			\notag  
			\lambda \  | \! 	 \< \Psi, \V^- \Psi  \> \! 		|
			& 
			\leq 
			\lambda N
			\textstyle 
			\sum_{k \in \Z^3}
			% \sum \mathop{}_{k \in \Z^3  } 
			|\widehat V(k)| 
			\sum_{\substack{ p\in L(k) }} \| c_{p- k} b_{p } \Psi \|  \\[1mm] 
			\notag
			& 
			\leq 
			N \ve    
			\< \Psi, \T  \Psi \> 
			\,   +  \,  
			C      \lambda^2 N  \ve^{-1 }
			\mathcal S_{1,2}(k_F)   \\[0.5mm]
			&  
			\leq 
			N \ve    
			\< \Psi, \T  \Psi \> 
			\,   +  \,  
			C 	 \lambda^2 N  \kf   \ve^{-1 }
			\|	V	\|_{H^4}  
			= 
			\delta 		 \< \Psi, \T  \Psi \> 
			+  C \delta^{-1 } N 		  \|	V	\|_{H^4}   \ .
			%		&  \ \leq  \ 
			%		N \ve     
			%		\< \Psi, \T  \Psi \>
			%		\,    + 		\,  C  \ve^{-1 } N  
			%		\| V  \|_{H^3}^2   \ . 
		\end{align}
		Recall 	 $ \lambda^2 N  \kf  =  4 \pi   $ is our choice of scaling. 
		Finally, for the bosonic Hamiltonian we 
		can easily deduce  from Lemma \ref{lemmaA1} that
		\begin{equation}
			\< \Psi, h \otimes \1 \Psi \>  \geq -  C N  \| W \|_{L^p}^{\frac{2p}{2p-3 }} 
		\end{equation}
		thanks to the non-negativity of the free kinetic energy. 
		We combine the last four inequalities  and re-define $\delta \mapsto \delta/10$  if necessary to finish the proof. 
	\end{proof}

	\section{Upper and lower bounds}\label{sec:effective}
	In this section, we complete the proof of the expansion 
	for the eigenvalues of $\mathbb H$.
	Let us recall that $(W,V)$ satisfy Condition \ref{Ass1} with power $ p \in ( \tfrac32, \infty] $.
	Throughout this section we will refer to  the factor  $\Q = \Q_p $ given by
	\begin{equation}\label{eq:Q}
		\Q = 1  + \| W \|_{L^p}^q + \| V\|_{H^4}^2 \ , \qquad q \equiv  \frac{2p}{2p-3}   .
	\end{equation}
	
	\subsection{Estimates for the upper bound\label{sec:upper:bound}}
	In this section, we prove the following proposition
	containing the  upper bounds.

	\begin{proposition}
		\label{prop2}
		There exists a constant  $C>0$  
		such that
		for all $ k_F \geq 1 $, $ N \geq 1 $ and $1  \leq n \leq N $ 
		\begin{align*}
			\mu_n( \H )  
			& 
			\, 		\leq  \, 
			\mu_n (h_{k_F}^\eff)  	\,	-	\,	
			\tfrac12  \| V \|_{L^2}^2 	\,	+	\,		 
			C \Q^2  N^2   ( \ln k_F) ^{5/3}  \kf^{-1/3}
		\end{align*}
		provided    
		$k_F^{1/3} \ge C  n  N \| V \|_{H^2}^2 (\ln k_F)^{5/3}.$
	\end{proposition}

	We will be using the following notation for the 
	free  fermionic resolvent
	\begin{equation}
		\RR  \equiv  \mathbb Q_ \Omega \mathbb{T}^{-1} \mathbb Q_ \Omega 
	\end{equation}
	with orthogonal projections $\mathbb Q_ \Omega = \1 - \mathbb P_\Omega $ and $\mathbb P_\Omega = |\Omega \rangle \langle \Omega | $.
	Note that thanks to Lemma \ref{lemma:kinetic:term},  we have 
	$ \mathbb Q _\Omega  \T \mathbb Q_\Omega \geq \mathbb Q_\Omega$.
	Thus, $\RR$ is well-defined, and  
	$ 0 \leq \RR\leq \1 . $
	
	\medskip 
	
	The following lemma will be useful to compute 
	the norm of relevant trial states. 
	Essentially, we gain  a     factor $\kf^{-1/3 }$
	relative to the leading order terms $\kf$
	thanks to the \textit{double} resolvent in the inner product.
	The latter can then be expressed a sum  over lune sets of the form analyzed in
	Section \ref{section:sums} with $\alpha =2 $. 
	
	\begin{lemma}
		\label{lemma:partial1} There is $C>0$ such that for all  
		$k_F\ge 1$,  $N\ge 1$ 
		and 
		$\Phi \in\mathscr H_B$ 
		\begin{equation}
			\label{lemma1eq1}
			\|	 \RR  \V_+ \Phi  \otimes \Omega 
			\|^2 
			\leq 
			C N^2  \| 	 V 	 \|_{H^2}^2 \,   \|  \Phi \|^2\,  k_F^{2/3} ( \ln k_F)^{5/3} \, \ . 
		\end{equation}
	\end{lemma}
	
	\begin{proof}
		From the definition   \eqref{V+}  
		and the 
		pull-through formula one obtains 
		\begin{align}
			\nonumber 
			\|	 \RR  \V_+ \Phi  \otimes \Omega 	\|^2
			=  
			\< 
			\Phi \otimes \Omega, 
			\V_- \RR^2 \V_+ 
			\Phi \otimes \Omega 
			\>
			=
			\sum_{1 \leq i,j \leq N } 	
			\sum_{ k \in \Z^3 }  
			| \widehat V (k)|^2 
			\sum_{ p \in L(k)}
			\frac{			\<	 \Phi 	, 	
				e^{ ik \cdot (x_i - x_j)}  
				\Phi  \>_{L^2 } }{ ( p^2 - (p - k)^2)^2} \ . 
		\end{align}
		Therefore, the Cauchy-Schwarz inequality readily implies 
		\begin{equation}
			\|	 \RR  \V_+ \Phi  \otimes \Omega 	\|^2 
			\leq N^2  	 \|		 \Phi	\|^2 	\sum_{ k \in \Z^3}
			\frac{ 		 | \widehat V(k)	|^2 }{	 (p^2 - (p - k)^2  )^2 }
			\ . 
		\end{equation}
		The sum over $ p \in L(k)$
		coincides with $\S_{2 ,0} (\kf)$
		and 	can   be   be estimated 
		with  Lemma \ref{lemmaS1}.
		This estimate gives the desired inequality. 
	\end{proof} 
	
	In the upcoming proof, 
	we let  $ \ho \equiv  \sum_{ i =1 }^N ( - \Delta_i)$
	be the bosonic kinetic energy on $\mathscr  H _B$. 
	We will use the following elementary bound
	\begin{equation}
		\label{eq:kinetic:bound}
		\mu_n (	 \ho 		) \leq  n 
	\end{equation}
	for  $1 \leq n \leq N $, 
	which follows from testing the kinetic energy operator 
	on the $n$-dimensional subspace  $\text{Span}(\phi_1, \ldots, \phi_{n}) \subset \mathscr H_{B}$ with $\phi_{j} = S_N \big(  e_0^{\otimes N-j} \otimes e_k^{\otimes j} )$, where $ k = (1, 0 , 0)$ and $S_N$ denotes the orthogonal projection on the symmetric subspace in $L^2(\To^{3N})$.
	
	\begin{proof}[Proof of Proposition \ref{prop2}]
		For $ n \in \N$ 
		we denote  by  $\Phi_n \in \mathscr H _{ B }$
		the (normalized) eigenfunction associated
		to $\mu_n( h_{ k_F}^\eff)$. 
		The proof is based on a suitable estimate 
		of the quadratic form of $\H$ with the trial state 
		\begin{equation}
			\label{psi}
			\Psi  = ( 1 - \lambda \RR \V_+ )
			\Psi_0 \ , 
			\qquad 
			\Psi_0
			\equiv  \Phi \otimes \Omega  
		\end{equation}
		where $\Phi  \in \t{Span} ( \Phi_1 , \cdots, \Phi_n) \subset \mathscr H_{B}$, 
		and we note that $\Psi \in \mathscr K$. 
		Using
		\begin{align}\label{eq:psi:H:psi}
			\< 		 	 \Psi , \mathbb H 	 	 \Psi \> 
			=
			\< 	 	 	 \Psi ,  h 	 	 \Psi  \>
			+ 
			\< 	 	 	 \Psi  , \T  		 	 \Psi \>
			+ 
			\lambda 	  \< 	 	 	 \Psi ,  \V  		 	 \Psi \> \ ,
		\end{align}
		we will  
		expand each energy term with respect to $\lambda$  in $\eqref{psi}$. 
		We will then 
		identify  the effective potential $W_{k_F}( x -y )$
		as an $O(\lambda^2)$  contribution, 
		and identify the error terms---to be estimated in the next step.
		A straightforward calculation shows 
		\begin{align}
			\label{eq4.3}
			\< 	 	 	 \Psi,  h 	 	 \Psi \>
			& 
			\,  = \, 
			\< 
			\Psi_0  , 
			h
			\Psi_0   \>
			\, 	 	 +  \, 
			\lambda^2 
			\<  		 	  \Psi_0 , 	 \V_-   (	h  \otimes \RR^2) \V_+ 
			\Psi_0 \>  \\ 
			\label{eq4.4}
			\< 	 	 	 \Psi , \T  		 	 \Psi\>
			& 
			\,  = \, 
			\lambda^2 
			\< 	 	  \Psi_0 , 	 \V_-  \RR 	  \V_+  	 	  \Psi_0\> 
			\\
			\label{eq4.5}
			\lambda 	  \< 	 	 	 \Psi,  \V  		 	 \Psi\>  
			& 
			\,  = \,  
			\, 	 	  -   \, 
			2	
			\lambda^2 
			\<  	 	  \Psi_0 , 	
			\V_-  \RR 	  \V_+  	 	  \Psi_0  \>  
			+  \lambda^3 
			\<	 	  \Psi_0  ,			 \V_- \RR 
			\V^{\t{diag}	}
			\RR \V_+		 	  \Psi_0  \>   \ . 
		\end{align}
		The following observation is crucial. 
		Namely,    from the  partial trace over $\t{Span}(\Omega)  \subset \mathscr H_{F}$ (i.e. the partial expectation value w.r.t to the   $\Omega$)  we obtain   the effective potential  \eqref{eq:def:eff:kF:potential}
		\begin{align}
			\notag 
			\lambda^2 
			\<	 \Phi \otimes \Omega ,  		 \V_-  \RR 	  \V_+ 		 \Phi \otimes \Omega  	 \> 
			&   \, =  \, 
			\< 
			\Phi \otimes \Omega , 
			\frac{1}{4 \pi N \kf}
			\sum_{1 \leq i , j \leq N }
			\sum_{ k \in \Z^3}
			\sum_{ p \in L(k)}
			\frac{ | \widehat  V (k)|^2 }{   p^2 - (p -k)^2  }
			\Phi \otimes \Omega
			\> \\ 
			&   \, =  \, 
			\<  \Phi, 	 \frac{1}{2N}
			\sum_{1 \leq i , j \leq N } W_{k_F}(x_i - x_j)				\Phi \>    \notag \\ 
			&     \, =  \, 
			\<  \Phi, 	 \frac{1}{N}
			\sum_{1 \leq i < j \leq N } W_{k_F}(x_i - x_j)				\Phi \>   \  +  \frac{1}{2}
			W_{k_F}(0)   \| \Phi\|^2  \ . 
		\end{align}
		Combined with \eqref{eq4.3}--\eqref{eq4.5} 
		and the definition of 
		$h_{\kf}^{\rm eff}$
		in \eqref{heff}, we obtain 
		\begin{align}
			\notag 
			\<
			\Psi  , \H \Psi 
			\>  
			&  	\ 	 =  \ 
			\<	 \Phi , h_{k_F}^\eff \Phi  	 \>
			\  -   \   \tfrac12 W_{\kf} (0 )		\| \Phi \|^2  \\
			& 	\qquad 	+
			\lambda^2 
			\<  	 	  \Psi_0, 	 \V_-  ( h \otimes  \RR^2)  \V_+ 
			\Psi_0 \> 
			\  + \  
			\lambda^3 
			\< 	 	  \Psi_0  ,			 \V_- \RR 
			\V^{\t{diag}	}
			\RR \V_+		 	  \Psi_0  \> 
			\notag
			\\
			\label{E0}
			& \  \equiv  \ 
			\<	 \Phi , h_{k_F}^\eff \Phi  	 \>
			\  -   \   \tfrac12 W_{\kf} (0 )		\| \Phi \|^2  \ + \  E_1  \ + \  E_2 \ . 
		\end{align}
		%	Using \eqref{supremum:difference} and $ (  V\ast V) (0)= \| V \|_{L^2}^2$ we can estimate 
		%	\begin{align} \label{self-energy-difference}
			%		\big| W_{k_F}(0) - \| V \|_{L^2}^2 \big| \le  \frac{C (\ln k_F)^{ 5/ 3  }}{ k_F^{1/3 }} 
			%		\|   V \|_{H^2}^2   \ .
			%	\end{align}
		We now turn to the estimates
		of the     two terms in \eqref{E0}, 
		to be denoted by $E_1$ and $E_2$. 
		
		\smallskip 
		
		\noindent  \underline{Estimates for $E_1$.}
		Thanks to  $W \in L^p(\To^3)$, 
		the operator $h$ can be bounded in terms of  the kinetic energy $\ho$. 
		That is, from Lemma \ref{lemmaA1} with $\ve = \frac{1}{2}$
		we obtain  for $q \equiv  \frac{2 p}{2p-3}$
		\begin{align}  
			\notag 
			E_1 
			& 	   \   =  \ 
			\lambda^2  
			\<  	 \RR \V_+    	  \Psi_0    , 	
			h    \ 
			\RR \V_+		 	  \Psi_0 \>  \\ 
			\notag 
			%		  \\
			%		&	
			&  \ 	   \leq 	 \ 
			C  \lambda^2 
			\(   
			\<  	 \RR \V_+    	  \Psi_0    , 	
			\ho       \ 
			\RR \V_+		 	  \Psi_0 \> 
			+ 
			N   \| W \|_{L^p}^q  
			\|	 \RR \V_+ \Psi_0	\|^2 
			\)  
			\\
			&   \ =   \ 
			C 		\lambda^2 N 
			\(   \, 
			\|  	\nabla_1 	\RR \V_+ 	 	  \Psi_0     \|^2  
			+ 
			\| W \|_{L^p}^q  
			\|  \RR \V_+	 	  \Psi_0   \|^2 
			\,  \)  \, . 
			\label{eq4.6}
		\end{align}
		where we used permutational symmetry 
		$   \<  \Psi, \ho \Psi  \>  = N  \|  \nabla_1  \Psi\|^2  $ in the second line. 
		The second term of \eqref{eq4.6} can be  readily estimated
		with Lemma \ref{lemma:partial1} 
		\begin{align}
			\|  \RR \V_+	 	  \Psi_0  \|^2 
			\, \leq  \, 
			CN^2
			(\ln \kf)^{5/3} \kf^{2/3}
			\| 	 V 	 \|_{H^2 }^2  \|  \Phi \|^2  \, . 
		\end{align}
		On the other hand, the first  term of \eqref{eq4.6}  
		can be estimated using 
		one commutator, 
		and   the triangle inequality 
		\begin{align}
			\label{eq4.7}
			\|  	\nabla_1 	\RR \V_+ 	 	  \Psi_0     \|^2 
			\leq 
			2  
			\|   	\RR \V_+	\nabla_1  	 	  \Psi_0     \|^2 
			+ 
			2  
			\|   	\RR  [ 	\nabla_1   , \V_+ ] 	 	  \Psi_0   \|^2    \ . 
		\end{align}
		These two terms must be analyzed separately.
		Essentially, thanks to the pressence of the two resolvents, we will gain a factor $\sim \kf^{-1/3}$
		relative to the leading order terms.
		
		\begin{itemize}[leftmargin= 0.6cm ]
			\item   
			The first term in \eqref{eq4.7}
			can be estimated   with Lemma \ref{lemma:partial1} 
			and gives rise to  
			\begin{align}
				\notag 
				\|   	\RR \V_+ (	\nabla_1 
				\Phi \otimes \Omega  )   \|^2 
				& 	 \leq  C  
				N^2 
				(\ln \kf)^{5/3}
				\kf^{2/3}
				\| 	 V 	 \|_{H^2}^2  \| \nabla_1 
				\Phi   \|^2 \,  \ . 
			\end{align}
			We estimate the gradient 
			using one more time Lemma \ref{lemmaA1} with $\ve =1 $
			and $  W_{\kf } \in L^\infty $  
			\begin{align}
				\notag 
				\|  \nabla_1 \Phi \|^2 
				\ 		=    \
				\frac{1}{N}
				\< \Phi	 ,	 \ho  \Phi  \>
				& 	 \ 	\leq  \ 
				\frac{2 }{N}
				\< \Phi	 ,	  h  \Phi  \>
				\ 			+ \ 
				C \| W  \|_{L^p}^q  \|   \Phi\|^2 \\
				&  \ 	 \leq   \ 
				\frac{C }{N}
				\<  \Phi, h_{ k_F}^\eff \Phi \>
				\ 	 		+  \
				C 
				\Big( 
				\|     W_{k_F}	\|_{L^\infty }   
				\ +  \ 
				\| W  \|_{L^p}^q \Big)   \|   \Phi\|^2  \  . 
				\label{mu1} 
			\end{align}
			Next, we use the fact that
			$\Phi \in  \t{Span} (\Phi_1 , \cdots, \Phi_n)$ with $\Phi_i$ the eigenfunction associated with the eigenvalue $\mu_i( h_{ k_F}^{\rm eff})$ to find 
			\begin{equation}
				\label{mu2}
				\frac{ \<  \Phi, h_{ k_F}^\eff \Phi \> }{\< \Phi, \Phi \>} 
				\leq 
				\mu_n ( h_{ k_F}^\eff	 )
				\leq 
				C 
				\mu_n (	  \ho  	)
				+  C N
				\Big(    \| W  \|_{L^p}^q  
				+        \|    W_{k_F}	\|_{L^\infty }  \Big)  
			\end{equation}
			where we employed Lemma \ref{lemmaA1} and  the min-max principle in the second inequality. 
			Thanks to \eqref{eq:kinetic:bound}	
			we find for $1 \leq n \leq N $ that 
			$ \mu_n ( \ho  	) \leq  N $.
			In addition, we recall that $W_{\kf} \in L^\infty$ is controlled 
			in  \eqref{eq:W:bound} by $ V \in H^2$. Thus, we get
			\begin{align}
				\label{eq:bound:ev:eff:h}
				\mu_n ( h_{ k_F}^\eff	 )
				\leq 
				C   N \(  
				1 +    \| W  \|_{L^p}^q  
				+         \|  V\|^2_{H^2}	  \)   
			\end{align}
			which then controls the gradient term $		\|   \nabla_1 \Phi\|^2 $. 
			
			\medskip

			\item 
			The second term in \eqref{eq4.7}
			contains a  commutator  that can be computed explicitly 
			\begin{align}
				[\nabla_1, \V_+]
				=  \sum_{p , k \in \Z^d }
				\widehat V ( k   )(-ik)
				e^{ - i  k  x_1 }
				\otimes   b_{p}^* c_{p + k  }^*  \ . 
			\end{align} 
			Doing the same calculation as in
			Lemma \ref{lemma:partial1}, we find thanks to 
			Lemma \ref{lemmaS1}  for $\alpha= 2$ and $\beta =1 $ that 
			\begin{align}
				\notag 
				\|	 \RR  [ \nabla_1 , \V_+  ] 
				\Phi \otimes \Omega 	\|^2  
				&  =   
				\Big\langle
				\Phi  	, 	\sum_{ k\in \Z^3 } |\widehat V (k)|^2 k^2 
				\D_2(k,k_F) 
				\Phi      \Big\rangle   	\leq  
				C  ( \ln \kf)^{5/3} \kf^{2 /3 }
				\|	V	\|_{	 H^3  	}^2  
				\|  \Phi  \|^2  \ . 
			\end{align} 
		\end{itemize}

		\noindent 	We put everything back in \eqref{eq4.6}
		and use $\lambda^2 N \sim \kf^{-1}$
		to find that there is a constant  $C_1>0$
		such that for all $ k_F \geq 1 $
		and $   1 \leq n \leq N   $
		\begin{align}
			\label{E1}
			E_1
			\leq 
			C N^2 
			(\ln \kf)^{5/3}
			\kf^{-1/3 }
			\bigg(  
			\| 	 V 	 \|_{H^2}^2   
			\big( 
			1 + 	\|  W\|_{L^p }^q
			+  
			\| 	 V 	 \|_{H^2 }^2    
			\big) 
			+   N^{-2 }
			\|	V	\|_{	 H^3	}^2   
			\bigg)  \| \Phi \|^2 \, .
		\end{align} 
		We note that the the term in brackets $( \cdots)$ is majorized by $\Q^2. $
		
		\medskip 
		
		\noindent  \underline{Estimates for $E_2$.}
		First,   in view of 
		Lemma \ref{lemma:diagonal:terms} 
		we  control $\V^{\rm diag}$ 
		in terms of $\calN_+$.
		Additionally,    note that   $	 \RR \V_+  (\Phi \otimes \Omega)$ 
		is an eigenvector of $\mathcal N _+$
		with eigenvalue 1.
		Thus,   we find that 
		for a constant $C> 0 $
		\begin{align}
			\notag
			E_2 
			&   \ = \ 
			\lambda^3 
			\<	\Psi_0   ,			 \V_- \RR 
			\V^{\t{diag}	}
			\RR \V_+		 	  \Psi_0   \>  \\ 
			%		\\
			%		\notag
			%		& 
			%		\,  \leq  \, 
			%		\lambda^3  
			%		N  
			%		\|  \widehat V \|_{\ell^1}
			%		\<		 	  \Psi_0  ,			 \V_- \RR 
			%		\, 2  \calN_+  \, 
			%		\RR \V_+			 	  \Psi_0  \>   \,  
			&   \ \leq      \ 
			C 	 \lambda^3  		 N 
			\|  \widehat V \|_{\ell^1}
			\|  	\RR \V_+			 	  \Psi_0   \|^2 
			\notag 
			\\ 
			& \   \leq  \ 
			C   N^{3/2 } 
			(\ln \kf)^{5/3} 
			\kf^{  - 5/6 }
			\|  \widehat V \|_{\ell^1} 
			\| 	 V 	 \|_{H^2   }^2     \| \Phi \|^2 \,  , 
			\label{E2}
		\end{align}
		where for  the
		second inequality 
		we used  Lemma \ref{lemma:partial1} 
		to evaluate
		$		\|  	\RR \V_+			 	  \Psi_0   \|^2 $, 
		as well as $\lambda^3 \sim N^{ - 3/2 } \kf^{-3/2}$. 
		Observe also that 
		$ \|   \widehat V \|_{\ell^1 } \leq C (1 + \|   V \|_{H^2  }^2		)$ and thus, 
		for $k_F \geq 1 $ and $ N \geq 1 $, 
		the right hand side of \eqref{E2}
		is bounded above  by the right hand side of  \eqref{E1}
		
		\medskip

		\noindent  \underline{Conclusion.}
		Going back to \eqref{E0},  we  obtain thanks to 
		\eqref{E1} and \eqref{E2} that
		there exists a constant $C>0$
		such that for all $ k_F  \geq 1 $
		and $ N \geq n  \geq1 $ 
		\begin{align}
			\label{eq4.8}
			\<
			\Psi , \H \Psi 
			\> 
			\leq  \big( \mu_n (		 h_{k_F}^\eff			)  - 
			\tfrac12 W_{\kf} (0 )	  \big) \| \Phi \|^2 	 
			+ 
			C  \Q^2  N^2 (\ln \kf)^{5/3} \kf^{ - 1/3 } 
			%		  \bigg(  
			%		\| 	 V 	 \|_{H^2  }^2   
			%		\big( 
			%		1+ 
			%		\|  W\|_{L^\infty }
			%		+  
			%		\| 	 V 	 \|_{H^2  }^2   
			%		\big) 
			%		+   N^{-1 }
			%		\|	V	\|_{	 H^3   	}^2   
			%		\bigg)  
			\| \Phi \|^2    \ , 
		\end{align}
		With this result in place, we   now employ the min-max principle   to find that 
		\begin{align}
			\notag
			\mu_n(\mathbb H) \le \sup_{ \Psi \in \mathcal V_n  } \frac{\langle \Psi , \mathbb H \Psi \rangle }{\langle  \Psi , \Psi \rangle } \, ,
		\end{align}
		where the supremum is taken over all states within the subspace
		\begin{align}
			\notag
			\mathcal V_n \equiv \t{Span}( \Psi_1,\ldots, \Psi_n ) , \quad 	 \Psi_j   =  (1 - \lambda \RR \V_+ ) (\Phi_j \otimes \Omega ) .
		\end{align}
		Below, we show that $\t{dim}(\mathcal V_n) = n $ for sufficiently large $k_F \geq 1 $. 
		Using \eqref{eq4.8} and the inequality $\|\Psi \|\ge \| \Phi \| $ for $\Psi =  (1 - \lambda \RR \V_+ ) (\Phi \otimes \Omega ) $ with $\Phi \in \t{Span}(\Phi_1,\ldots, \Phi_n)$, we derive the bound
		\begin{equation}
			\notag
			\mu_n(\mathbb H) \le  \mu_n (		 h_{k_F}^\eff			) - 	  \tfrac12 W_{\kf} (0 )	
			+
			C   \Q^2    N^2 (\ln \kf)^{5/3} \kf^{ - 1/3 }  \ . 
		\end{equation}
		
		\smallskip
		So far, we have not made any further restrictions on $ \kf \geq 1 $ and $ N \geq 1 $.
		For the next argument, consider  a large enough  constant $R>0$ and take $\kf \geq 1 $ large enough so that 
		\begin{equation}
			\label{eq:R}
			k_F^{1/3} \ge  R   n N \| V \|_{H^2}^2 (\ln k_F)^{5/3} \ .
		\end{equation}
		Let us now  establish that $\mathcal V_n$ has dimension $n$.
		Indeed,  observe that $\langle \Phi_i,\Phi_j \rangle = \delta_{i,j}$ by assumption. Moreover, using Lemma \ref{lemma:partial1}, we find for $i\neq j$:
		\begin{align}\label{eq:Psi:i:Psi:j:bound}
			| \langle \Psi_i, \Psi_j\rangle  | &  \ \le \   \lambda^2 \| \RR \V^+  \Phi_i \otimes \Omega \| \, \| \RR \V^+  \Phi_j \otimes \Omega \|  \  \le \  C_0 N \| 	 V 	 \|_{H^2}^2 
			(\ln \kf)^{5/3} \kf^{ - 1/3 } , 
		\end{align}
		for some constant $C_0$. 
		If $\t{dim}(\mathcal V_n) < n$, it would imply that the set $\{ \Psi_i\}_{i=1}^n$ is linearly dependent. Consequently there would exist an integer $m \in \{1,\ldots n\}$ such that
		\begin{align}
			\notag
			\Psi_m = 
			\sum_{ j=1: j\neq m }^n \alpha_j 
			\Psi_j \quad \text{with}
			\quad \alpha_j \in    {\bold C }  , \quad |\alpha_j|\le 1.
		\end{align}
		Taking   the inner product with $\Psi_m$    and applying \eqref{eq:Psi:i:Psi:j:bound}, we find
		\begin{align}
			\notag
			\| \Psi_m \|^2 \le C_0 (n-1) N \| 	 V 	 \|_{H^2}^2 
			(\ln \kf)^{5/3} \kf^{ - 1/3 }
			< C_0/ R < 1 
		\end{align}
		provided $R >C_0$. However, this contradicts $\| \Psi _m \| \ge  1$, and thus, we conclude that $\t{dim}(\mathcal V_n) = n$.  
		This completes the proof of Proposition \ref{prop2}. 
	\end{proof}

	\subsection{Estimates for  the lower bound\label{sec:lower:bound}}
	
	The purpose of this section is to prove 
	the following operator inequality, 
	which implies  
	the lower bound in Theorem \ref{thm1}. 
	Let us recall that $(W,V)$ satisfy Condition \ref{Ass1} with power $ p \in ( \tfrac32, \infty] .$
	For the next statement, we denote
	\begin{equation}
		\epsilon_0   
		=			N 			( \ln \kf)^{ 5/3} \kf^{ -1 /3 } \ .
	\end{equation}
	We also recall $\Q  = 1 + \| W \|_{L^p}^q + \| V\|_{H^4}^2$ with $q \equiv  \frac{2p}{2p-3}$. 
	
	\begin{proposition}\label{prop1}
		There is $C>0$  
		such that 
		for all $ k_F \geq 1 $ and $ N \geq 1 $
		the following holds
		\begin{enumerate}[leftmargin=*]
			\item  	In the sense of quadratic forms  on	  $\K = \ker \calN_-$
			\begin{align}
				\label{eq:op}
				\H
				\, 	\geq  \, 
				(1 -  C  \Q    \epsilon_0     )	
				\Big(     		h_{ k_F}^\eff   
				-   \tfrac12 W_{\kf}(0)
				\Big)   -  
				C    \Q 
				N^2 
				( \ln \kf)^{ 5/3} \kf^{ -1 /3 }
				%		\quad  \text{with} \quad 
				%		\epsilon_0   
				%		\, =  \, 
				%		N 
				%		( \ln \kf)^{ 5/3} \kf^{ -1 /3 }
			\end{align}
			provided  $ k_F (\ln k_F)^{-5 } \ge C     		\|	V	\|_{H^4}^6     N^3$. \\[-3mm]
			
			\item  
			Let  $  1 \leq n \leq N $. 
			Then, it holds that
			\begin{align}
				\mu_n(\H) 
				\, 	\geq  \, 
				\mu_n(	 h_{k_F}^\eff	)  
				-   \tfrac12 W_{\kf}(0)
				-  	 C \Q^2   N^2 \,  
				( \ln \kf)^{ 5/3} \kf^{ -1 /3 } 
			\end{align}
			provided  $ k_F (\ln k_F)^{-5 } \ge C     		\|	V	\|_{H^4}^6     N^3$. 
		\end{enumerate}

	\end{proposition}

	\begin{proof}
		Throughout this proof, all operator inequalities are to be understood
		in the sense of quadratic forms on	  $\K = \ker \calN_-$. 
		\smallskip  
		
		(1)	First, we remove the diagonal term $\V^{\rm diag}$
		by utilizing Lemmas \ref{lemma:kinetic:term} and \ref{lemma:diagonal:terms}
		\begin{align}
			\label{eq:1}
			\mathbb H 
			\, \ge	\,
			h + \mathbb T  + \lambda  \mathbb V^+ +  \lambda \mathbb V^-   -
			{C N^{  1/2 }}	\kf^{-1/2 }	 \|  \widehat V\|_{\ell^1 }   \mathbb T   \ . 
		\end{align}
		Thus, it is  sufficient 
		to analyze the  first three terms on the right hand side.
		By completing the square, we find
		the operator inequality   
		\begin{align}
			\label{eq:2}
			\mathbb T + \lambda \mathbb V^- + \lambda \mathbb V^+  & =     \Big( \mathbb T^{1/2} + \lambda \mathbb T^{-1/2}  \mathbb V^+ \Big)^*\Big( \mathbb T^{1/2} + \lambda\mathbb T^{-1/2}  \mathbb V^+ \Big) - \lambda^2 \mathbb V^- \mathbb T^{-1} \V^+ \notag\\
			& \ge -\lambda^2 \mathbb V^-  \mathbb T^{-1}  \V^+.
		\end{align}
		Note  the above expressions are well-defined.
		This follows from
		$  \V_+  \K   \subset  \1 (\mathcal N_+ \geq 1 ) \K $
		and the lower bound 
		$   \mathbb T \mathds 1(\mathcal N_+ \ge 1) \ge   \mathds 1(\mathcal N_+ \ge 1) $ .

		Next, we compute
		\begin{align}
			\label{eq4.9}
			&   \mathbb V^- \mathbb T^{-1}  \V^+  
			\   = \ 
			\sum_{i,j=1}^N \sum_{ k , \ell   \in \Z^3 } \overline{v_i( k   )}   v_j(\ell ) 
			\,  \sum_{p ,q \in \mathbb Z^3}     b_p c_{p -k} \mathbb T^{-1} c_{q - \ell }^* b_{q}^* ,
		\end{align}
		where we introduced $v_i(k) = e^{-ikx_i} \widehat V(k) $ and we keep in mind the implicit restrictions $p,q\in B^c$ and $p -k,q - \ell \in B$ coming from 
		$\{ c_p\}$ and  $ \{  b_p \}$.
		In particular, we can always assume $k$ and $\ell$ are non-zero. 
		Let us denote $E_p  \equiv   p^2 $. 
		We then use the pull-through formula
		\begin{align}
			\label{eq4.10}
			b_p c_{p - k} \mathbb T^{-1} c_{q - \ell }^* b_{q}^* 
			\  =  
			\  (   \mathbb 
			T +  E_p  -   E_{p-k}
			)^{-1/2} 
			\  b_p c_{p- k} c_{q - \ell }^* b_q^*  \  (   \mathbb T +E_q -  E_{q-\ell  })^{-1/2}
		\end{align}
		and apply the  CAR  to find
		\begin{align}
			\label{eq4.11}
			& b_p c_{p - k} c_{q  -  \ell}^* b_q^*  \\
			& 
			\, = \, 
			\delta_{p,q}   \delta_{k,\ell}  \chi^\perp(p) \chi(p+k) 
			\,  -  \, 
			\delta_{p,q} \chi^\perp(p)  c^*_{q - \ell} c_{p- k}   
			\, - \, 
			\delta_{p - k , q  - \ell} \chi(p - k) b_q^* b_p 
			\, + \, 
			b_q^* c_{q - \ell}^* c_{p -k} b_p \ . 
			\notag 
		\end{align}
		Upon combining \eqref{eq4.9}--\eqref{eq4.11}
		we obtain   
		$  \V^- \mathbb T^{-1}  \V^+  \equiv \sum_{n=1}^4 A_n$, 
		where: 
		\begin{align}
			A_1 		 &  	 = 
			\ 	 \sum_{i, j}
			\sum_{ k \in \Z^3   }
			\ 
			\overline{v_i (k) }  v_j(k)  
			\sum_{  p\in L(k)  } (   \mathbb T +  E_p - E_{p-k}  )^{-1 }
			\\
			A_2 		 &  	 = 
			- 
			\ 	 \sum_{i, j} \sum_{ k ,\ell    }  \ 
			\overline{v_i(k)} v_j(\ell)   \sum_{p\in B^c} 
			\  (   \mathbb 
			T +  E_p  -   E_{p-k}
			)^{- \frac{1}{2}} 
			c^*_{p - \ell} c_{p- k}    
			\  (   \mathbb T +E_p -  E_{p-\ell }  )^{- \frac{1}{2}}      
			\\
			A_3 		 &  	 = 		  - 
			\ 	 \sum_{i, j} \sum_{ k ,\ell    }  \  \overline{v_i(k)} v_j(\ell)   \sum_{p-k\in B} (\mathbb T +  E_p - E_{p - k}  )^{- \frac{1}{2}}  b_{p - k + \ell}^*  b_p (\mathbb T + E_{p  -  k + \ell } - E_{p - k}  )^{- \frac{1}{2}}  
			\\ 
			A_4  		 &  		 = 
			\ 	 \sum_{i, j} \sum_{ k ,\ell    }  \  \overline{v_i(k)}  v_j(\ell) \sum_{p,q  \in \Z^3 }  (\mathbb 
			T + E_p  -   E_{p-k}  )^{- \frac{1}{2}} b_q^* c_{q - \ell}^* c_{p - k} b_p (\mathbb T +
			E_q  -   E_{q -\ell }
			)^{- \frac{1}{2}} 
		\end{align}
		where $\sum_{i,j} \equiv \sum_{ i , j =1}^N$
		and $\sum_{ k , \ell} \equiv \sum_{ k ,\ell \in \Z^3}$
		and we recall  
		$L(k) = \{   p : |p|  >  k_F ,  | p - k|  \leq  k_F  		\}. $
		We now study each term separately. 
		
		\begin{itemize}[leftmargin=*]
			\item  
			The first term  $A_1 $ gives the effective potential \eqref{eq:def:eff:kF:potential}.
			Namely, using  
			$\mathbb T\ge 0$ on  $ \K  $ we find   
			\begin{align}
				\label{A1}
				\lambda^2 A_1 
				&  \ \le \   \lambda^2 \sum_{i,j=1}^N 
				\sum_{ k  \in \Z^3    }  | \widehat V ( k   )|^2  e^{ i k (x_i - x_j) }   \sum_{  p \in L(k)  } \frac{1  }{p^2 - (p- k)^2  } \notag\\
				&  \ = \  \frac{1}{2N}
				\sum_{1 \leq i, j \leq N }W_{k_F}(x_ i -x_j) 
				\ = \  \frac{1}{N}
				\sum_{1 \leq i < j \leq N }W_{k_F}(x_ i -x_j)  
				+ \frac{1}{2}W_{k_F}(0).
			\end{align}
			
			\smallskip
			\item 
			The second and third terms are non-positive. 
			To see this,  we use the representation 
			\begin{align}\label{rep}
				(\mathbb T + E_p + E_q - E_{p ' }- E_{q '  })^{-1}  
				=
				\int_0^\infty 
				e^{-s( E_q -  E_{q  ' })} e^{-s \mathbb T}
				e^{-s( E_p -   E_{p  '  })} \, ds
			\end{align}
			for any $ p , p', q , q ' \in \Z^3$.
			Indeed, using the pull-through formula  we obtain 
			\begin{align}
				\notag
				A_2 & =   - 
				\sum_{i,j=1}^N 
				\sum_{ k ,\ell \in \Z^3    }  \overline{v_i(k)} v_j(\ell)   \sum_{p\in B^c} 
				c^*_{p - \ell} (\mathbb T + E_p  -  E_{p - \ell } - 
				E_{p - k })^{-1} c_{p - k}  \notag\\
				\notag
				& 	=		- 
				\sum_{p\in B^c} 
				\int_0^\infty ds 
				\bigg( \sum_{j=1}^N \sum_{  \ell \in  \Z^3  }  \overline{v_j(\ell)}
				e^{s E_{p - \ell } }  c_{p - \ell} \bigg)^* e^{-s (\mathbb T +E_p  )} 
				\bigg( \sum_{i=1}^N \sum_{k\in  \Z^3 } \overline{v_i(k)} e^{s E_{p - k }} c_{p - k}  \bigg) \\
				& \leq 0  \ . 
				\label{A2}
			\end{align}
			Note  the first line is well-defined, since 
			$\mathbb T - E_{p - k }$ 
			takes only non-negative values on states of the form $c_{p-k}\Psi$ with $\Psi \in   \K $, while $E_p - E_{p - \ell }$ is strictly positive.
			As for the third term, we use  the CAR and a change of variables
			$ p \mapsto  p   + k $ \allowdisplaybreaks
			\begin{align}
				\notag
				A_3  
				& =  - 
				\sum_{i,j=1}^N \sum_{ k ,\ell \in \Z^3    } \overline{v_i(k)} v_j(\ell)    \sum_{p - k\in B}  b_{p- k+ \ell}^* (\mathbb T + E_{p - k  + \ell }  +  E_p - E_{p - k} )^{-1}    b_p   \notag\\
				\notag
				&  =  - 
				\sum_{i,j=1}^N\sum_{ k ,\ell \in \Z^3    } \overline{v_i(k)} v_j(\ell)  
				\sum_{p\in B}  b_{p+ \ell}^* 
				(\mathbb T + E_{p +\ell }+  E_{p+k} - E_p )^{-1}    b_{p+ k}   \\
				&  = - 
				\sum_{p\in B} \int_0^\infty ds 
				\bigg( \sum_{j=1}^N \sum_{\ell\in  \Z^3  } e^{-s E_{p-\ell}}
				\overline{v_j(\ell)} b_{p + \ell} \bigg)^*
				e^{-s(\mathbb T - E_p ) }  \bigg( \sum_{i=1}^N \sum_{k\in  \Z^3 } e^{-s E_{p + k }} \overline{v_i(k)}b_{p+ k} \bigg) \notag\\
				& \leq 0  \ . 
				\label{A3}
			\end{align} 
			Similarly,    the first line is well-defined as 
			$\mathbb T+ E_{p+k }$  is non-zero
			on states  $b_{p+ k}\Psi$ with $\Psi\in  \K  $, while 
			$E_{p + \ell } - E_p $ is strictly positive.
			
			\smallskip
			
			\item  For the fourth term, we proceed as follows.
			We use the CAR 
			to obtain 
			\begin{align}
				A_4   
				& = \sum_{i,j=1}^N
				\sum_{ k ,\ell \in \Z^3    } 
				\overline{v_i(k)}  v_j(\ell)   \sum_{p,q  \in \Z^3 } b_q^* c_{q - \ell}^*
				(\mathbb T + E_p  +E_q -  E_{p -k }- E_{q- \ell } )^{-1}  
				c_{p- k} b_p  \notag\\
				& \le \sum_{i,j=1}^N  \sum_{ k ,\ell \in \Z^3    } \overline{v_i(k)}  v_j(\ell)   \sum_{p,q  \in \Z^3 }  
				b_q^* c_{q - \ell}^*  
				(   E_p  +E_q -  E_{p -k }- E_{q- \ell } )^{-1}  
				c_{p- k} b_p
			\end{align}
			where the last inequality  follows from    
			\eqref{rep}
			and  $e^{-s\mathbb T } \le 1$ on 
			$\K $. 
			For the next step, 
			let us note that 
			$$
			(   E_p  +E_q -  E_{p -k }- E_{q- \ell } )^{-1}  
			\leq 
			(   E_p   - E_{p -k }  )^{-1/2 }  
			(   E_q  - E_{q- \ell } )^{-1/2 }.
			$$
			Let us now evaluate $A_4$ in a 
			normalized state  $\Psi  \in 
			\K $.
			We use the Cauchy--Schwarz inequality, 
			Lemma \ref{lemma:HPR} for
			$\alpha=2$, $\ve >0$ 
			and Lemma \ref{lemmaS1} for $\S_{2,2}(k_F )$ to find 
			\begin{align}
				\notag 
				\lambda^2 \, 
				\langle \Psi ,    A_4 \Psi \rangle    &   
				\ \le 	\ 
				N^2 \bigg( \lambda \sum_{k\in \Z^3 } |\widehat V(k)| \sum_{p\in L(k) } 
				( E_p -E_{p -k})^{-1/2}
				\| c_{p - k} b_p \Psi \| \bigg)^2    \\
				% &     \leq 
				% C N^2   
				% \bigg(
				% \ve    \| \mathbb T^{1/2} \Psi \|^2 +  
				% \ve^{-1 }
				% 	\lambda^2  \mathcal S_2(k_F)  
				% \bigg)^2	\\	
				& \   \leq  \ 
				N^2   
				\bigg(
				\ve   
				\< \Psi, \T \Psi \>
				+  C
				\ve^{-1 }N^{-1 }k_F^{ -1 /3 }
				( \ln k_F)^{5/3}		 		\|	V	\|_{H^4}^2  
				\bigg)^2 \ . 
			\end{align}
			% Note that  we used $b_p \equiv 0$ if $|p| < k_F$ and $c_{p-k} \equiv 0$ if $|p-k|>k_F$ in the first line. 
			Finally, we minimize with respect to $\ve>0$
			%We then choose 
			%\begin{align}
			%\ve^{-1} =  
			%\frac{N^{1/2}   k_F^{1/6}}{		\|	V	\|_{H^4}  ( \ln k_F)^{5/6}		 }   \| \mathbb T^{1/2} \Psi \| ,
			%\end{align}
			to find the estimate 
			\begin{align}
				\lambda^2 \, 
				\langle \Psi ,    A_4 \Psi \rangle  
				\le  
				C N 
				\|	V	\|_{H^4}^2  
				k_F^{ - 1/ 3}
				( \ln k_F)^{5/3}		 
				\< \Psi, \T \Psi \> \ . \label{A4}
			\end{align}
			
		\end{itemize}
		
		\medskip 
		
		\textit{Putting everything together}. 	We deduce from   \eqref{eq:1}, 
		the completion of the square \eqref{eq:2}, 
		and the estimates \eqref{A1}, \eqref{A2}, \eqref{A3} and \eqref{A4}
		that 
		there is a constant   $C >0$ such that 
		for all  $ N  \geq1  $ and $k_F \geq 1$
		\begin{equation}
			\label{eq:3}
			\H \ \geq \  h_{ k_F}^{\eff} 
			-   \tfrac12 W_{k_F}(0)-
			C_0
			\|	V	\|_{H^4}^2  
			\epsilon_0
			\T \ 
			\quad \text{with} \quad 
			\epsilon_0  =  
			N  ( \ln k_F)^{5/3}  k_F^{-1/3 }  \ . 
		\end{equation}
		Finally, it suffices to  bound the kinetic energy.
		To this end,  we  apply Lemma
		\ref{lemma:kinetic1} 
		with $\delta = 1 $ to see that 
		there is $C>0$  
		such that for all 
		$  N \geq 1  $
		and $k_F \geq C  N \| \widehat V\|_{\ell^1}^2 $ 
		\begin{align}
			\label{eq:4}
			\mathbb T   
			\  \leq  \ 2  \, 
			\mathbb H   
			\, +  \, 
			C  N 	\Q  \  . 
		\end{align}
		We combine the last two inequalities to obtain 
		\begin{equation}
			\( 1 + 	    	 C_0 	\|	V	\|_{H^4}^2    \epsilon_0  \) 
			\H 
			\ \geq  \ 
			h_{\kf}^\eff  	\	-	\	 \tfrac12 W_{\kf} (0)
			\	   -	\
			C \Q  N  \epsilon_0   
		\end{equation}
		for some distinguished constant $C_0$. 
		Assume   now $k_F \geq 1 $ is large enough 
		so that
		$$  0 <  1  -    C_0 		\|	V	\|_{H^4}^2    \epsilon_0  <  
		(1  +    C_0 		\|	V	\|_{H^4}^2   \epsilon_0)^{ -1  } <   1 \ .  $$ 
		Thus we get 
		from \eqref{eq:3}, \eqref{eq:4}  
		and bootstrapping 
		\begin{equation}
			\H 
			\ \geq \ 
			(1  -  C _0	\|	V	\|_{H^4}^2      \epsilon_0) 
			\(    h_{ k_F}^{\eff} 
			-  \tfrac12 W_{k_F}(0) 
			\) 
			- 
			C \Q   N \epsilon_0  
		\end{equation}
		which finishes the proof of the first statement. 
		
		\medskip 	
		
		(2)	 Recall that from \eqref{eq:W:bound} we
		have  
		$ \|  W_{\kf}\|_{L^\infty} \leq C \| V\|_{H^2}^2 \leq C \Q$.
		In addition, from 
		\eqref{eq:bound:ev:eff:h} 
		we have the bound    $\mu_n (h^\eff_{\kf}) \leq C  N \Q $
		for any $  1 \leq n \leq N $. 
		Thus, we can now employ the min-max principle to find   
		for all $1 \leq  n \leq  N$ 
		(and   all $\kf$ sufficiently large)
		\begin{align}
			\notag 
			\mu_n(\H) &
			\  \geq  \ 
			(1  -   C _0	 \Q     \epsilon_0 ) 
			\( 
			\mu_n( h_{ k_F}^{\eff} )	  	 	-  \tfrac12 W_{k_F}(0) 
			\) 
			-   	  C  \Q   N \epsilon_0  	\\
			& 	 \ 	\ge  \ 
			\mu_n( h_{ k_F}^{\eff})  	 	-  \tfrac12 W_{k_F}(0)      - C  \Q^2  N 
			\epsilon_0 \ . 
		\end{align}
		This completes the proof of the proposition.
	\end{proof}

	\section{Lattice sums}
	\label{section:sums}
	The main purpose
	of this section
	is to complete the proof of Lemma \ref{prop:sums}.
	In order to prove the statement, we shall 
	borrow some results from
	\cite{Christiansen2023}. We start by recalling the definition of the lune set
	\begin{equation}
		L(k) = \{  p \in \Z^3 :  |p | > k_F \ \t{and} \ | p - k | \leq  k_F 		 \} \ , \qquad  k \in \Z^3  , 
	\end{equation}
	as well as the relevant singular denominators
	\begin{equation}
		\lambda_{ p ,k } 
		\equiv \frac{1}{2} \big(	 |p|^2 - | p -k|^2				\big) \ , \qquad p , k \in \Z^3 \ . 
	\end{equation}
	Then, for any function $f : ( 0 , \infty) \rightarrow  ( 0  , \infty )  $
	the following \textit{summation formula} holds. 
	
	\begin{proposition}[\cite{Christiansen2023} Proposition B.20]
		For all $ k = (k_1 , k_2, k_3) \in \Z^3$
		with $| k |  < 2 k_F $
		let us denote: 
		\begin{itemize}
			\item $\ell\in (  0 , \infty )$ as the positive number 
			$\ell  =  |k|^{-1} \t{gcd} (k_1, k_2 , k_3)$. 
			\item $m_* \in \Z $ as the least integer such that: 
			$ \frac{|k|}{2} < \ell m_*$
			
			\item $M \in \Z $ as the greatest integer 
			such that:  $\ell M   \leq k_F$
			
			\item $M_* \in \Z $
			as the greatest integer such that: 
			$\ell M_* \leq k_F + |k|$
		\end{itemize}
		Then, as $k_F \rightarrow \infty$
		\begin{align}
			\label{sum:eq1}
			\sum_{ p \in L (k)	} f(	 \lambda_{ p , k}  	)
			&  \ 	=  \ 
			2 \pi |k|
			\sum_{	 m_* \leq m \leq M }
			f\bigg(	 |k| (\ell m - \frac{| k |}{2 })			\bigg)	
			\Big(	 \ell m - \frac{| k |}{2}			\Big) \ell 	\\ 
			\label{sum:eq2}
			&   \ \ +  \ 
			\pi \sum_{	 M+1 \leq  m\leq M_* 		}
			f\bigg(	 |k | (\ell m - \frac{| k |}{2 })				\bigg)
			\Big(			 k_F^2 - ( \ell m - k 	)^2				\Big)\ell  \\
			\label{sum:eq3}
			&  \ \ +  \ 
			O 	  \bigg(
			|k|^{  11/3 }  
			\ln(k_F)  k_F^\frac{2}{3}
			\sum_{ m_* \leq m \leq M_* 		}
			f  \Big(   | k |  (\ell  m - \frac{|k |}{2 })		 \Big) 	 
			\bigg)
		\end{align}

	\end{proposition}
	
	\begin{remark}
		\label{remark:numbers}
		The following observations follow   from the definition
		of the parameters $\ell$, $m_*$, $M$  and $M_*$
		and will become helpful in the upcoming calculations: 
		\begin{equation}
			\ell \leq 1   \ ,  \
			\qquad 
			\ell m_* \leq \frac{|k|}{2} + \ell 
			\qquad
			\ell m_* - \frac{ |k|}{2} \geq \frac{1}{2 |k|}  \ , \ 
			\qquad  
			k_F - \ell 
			\leq \ell M \ , \ 
			\qquad 
			k_F +  |k | - \ell \leq \ell M_* \  .
		\end{equation}
		All 
		are easy to obtain except the third one.
		It follows from  
		$2  \t{gcd} ( k_1, k_2 , k_3)   m_*  \geq |k|^2 + 1 $. 
	\end{remark}
	
	The next proof follows closely that of 
	\cite[Proposition B.21]{Christiansen2023}. 
	
	\begin{proof}[Proof of Lemma \ref{prop:sums}] (Inequality \eqref{eq:bound:S2})
		First, we observe that after a change of variables $ p \mapsto p - k$
		\begin{equation}
			\sum_{ p\in \Z^3 }
			\frac{		\chi^\perp (p) \chi (p+k)  		}{(E_p - E_{p+k})^{2 }		}
			= 
			C
			\sum_{ p \in L (k)  }  \lambda_{p,k}^{-2 }
		\end{equation}	
		where the Fermi momentum on the right
		is given by $k_F  '  =  2 \pi k_F$. 
		As in the final estimate the $2 \pi$ factors will play no role
		we shall ignore it and estimate the right hand side as
		for a given $k_F$. 
		We now use the Proposition with $f(\lambda) = \lambda^{-2}$. 
		
		The first term in \eqref{sum:eq1} can be estimated as follows 
		using 
		the integral estimate
		$ \sum_{ n=a}^b f(n) \leq  f(b)  + \int_{a}^b f(x) dx $
		for decreasing functions 
		\begin{align}
			\nonumber
			T_1 
			&  \ \equiv  \ 
			2\pi |k |
			\sum_{	 m_* \leq m \leq M }
			\frac{		 (\ell m - \frac{|k|}{2}  )	\ell 	}{  |k|^2 
				(\ell m - \frac{|k|}{2}  )^2	} \\
			\nonumber
			& \  =  \ 
			\frac{2\pi \ell }{| k  |}
			\sum_{	 m_* \leq m \leq M }
			\frac{1}{(\ell  m - \frac{|k|}{2}) } 
			\\ 
			\nonumber
			& \  \leq  \ 
			\frac{2\pi \ell }{| k  |} 
			\bigg(
			\frac{1}{ (  \ell M  - \frac{|k|}{2 } )  }
			+ 
			\int_{ m_*	}^M 
			\frac{\d m }{(\ell  m - \frac{|k|}{2})  } 
			\bigg)  \\
			& 
			\ 	=   \ 
			\frac{2\pi \ell}{| k |}
			\bigg(
			\frac{1}{ (  \ell M  - \frac{|k|}{2 } )  }
			+  
			\ln( \ell M  - \frac{ | k| }{2  })
			- 
			\ln( \ell m_*  - \frac{|k|}{2 } )	 
			\bigg)  \ . 
		\end{align}  
		Next, we use the various inequalities from Remark \ref{remark:numbers} and obtain 
		\begin{align}
			\nonumber 
			T_1  
			\ 		\leq  \ 
			\frac{2\pi \ell}{| k |}
			\bigg(
			\frac{1}{ ( k_F  - 1  - \frac{| k |}{2} ) 	}
			+  
			\ln(  k_F  - \frac{ | k| }{2  })
			- 
			\ln( \ell  )	 
			\bigg)   
			\ \leq  \ 
			\frac{C}{|k|}
			\bigg(
			k_F^{-1 }
			+  
			\ln k_F
			+ C_*  
			\bigg)   
		\end{align}  
		where in the last inequality we used $\ell \leq 1 $
		and $ - \ell \ln \ell \leq C_*$. 
		Thus $T_1 \leq \frac{C}{|k| } (1 + \ln k_F)$.

		\vspace{1mm}
		Arguing similarly, we now estimate the second error 
		term in \eqref{sum:eq2} as follows 
		\begin{align}
			\nonumber
			T_2 
			& \equiv 
			\pi |k |
			\sum_{	 M+1 \leq  m\leq M_* 		}
			\frac{			( 		
				k_F^2 - ( \ell m - k	)^2 			) \ell 	}{  |k|^2 
				(\ell m - \frac{|k|}{2}  )^2	} \\
			\nonumber
			& \leq 
			\pi | k |
			\sum_{	 M+1 \leq  m\leq M_* 		}
			\frac{		
				2 |k |(\ell m - \frac{| k |}{2})
			}{  |k|^2 
				(\ell m - \frac{|k|}{2}  )^2	}  
			\\ 
			\nonumber
			&  =  
			2 	\pi  
			\sum_{	 M+1 \leq  m\leq M_* 		}
			\frac{		
				1 
			}{   
				(\ell m - \frac{|k|}{2}  ) 	}  
			\\ 
			\nonumber
			& \leq 
			2\pi
			\bigg(		
			\frac{1 }{ \ell M_* -   |k|/2 }
			+ \ln\Big(	 \ell M_* - | k  |/2 	\Big)
			- 
			\ln\Big(	 \ell ( M +1 ) -  |k|/2 	\Big)
			\bigg) \\ 
			& \leq 
			C  \Big(  
			k_F^{- 1 } + \ln k_F 
			\Big) \ . 
		\end{align}
		Note that in the first line we used
		$ k_F^2 \leq (\ell m )^2	$ for $m \geq M +1$. 
		Thus  $T_2 \leq C  (1 + \ln k_F)$
		
		\vspace{1mm}
		
		Let us now turn to the third error term in \eqref{sum:eq3}.  To this end,  we estimate the sum 
		using the same methods as before. We obtain 
		thanks to the bounds in Remark \ref{remark:numbers}
		\begin{align}
			\nonumber
			T_3 &  \equiv 		\sum_{	 m_* \leq m \leq M 	}
			\frac{1}{(|k|  (\ell m - \frac{|k|}{2}))^2 } \\
			\nonumber
			& 
			\leq 
			\frac{1}{|k|^2 }
			\bigg(
			\frac{1}{ 		(\ell M_* - \frac{|k|}{2})^2}
			+ 
			\frac{1}{\ell} 
			\Big(
			\frac{1}{\ell m_*  -		 \frac{|k|}{2}}
			-
			\frac{1}{\ell M   -		 \frac{|k|}{2}}
			\Big)
			\bigg) \\
			& 
			\leq 
			\frac{C}{|k|^2}
			\bigg(
			k_F^{-2} 
			+ 
			\frac{ | k |}{ \ell}
			\bigg) 
			\leq 
			\frac{C}{|k|^2}
			\bigg(
			k_F^{-2} 
			+ 
			|k|^2 
			\bigg) 
		\end{align}
		where in the last line we used $\ell |k| = \gcd (k_1, k_2,k_3)\geq 1 $. 
		Thus, $T_3 \leq C $. 
		It suffices now to combine the estimates together.
	\end{proof}

	\appendix 
	
	\section{The bosonic quadratic form}
	Let $h_0  \equiv \sum_{ i =1 }^N (-\Delta_{x_i} )$ 
	be the free operator with form domain
	$ Q( h_0 )=   H^1 (\To^{3N }) \cap \mathscr H_B$, 
	and let
	$h  = h_0   + N^{ -1 } \sum_{ i< j }  w(x_i - x_j)$
	be the bosonic Hamiltonian, 
	defined initially on $C^\infty(\To^3) \cap \mathscr H_B$.

	\begin{lemma}
		\label{lemmaA1}
		Assume   $w   \in L^p  (\To^3)$ for some $3/2 < p \leq \infty$. 
		Then, 
		there are  $a, c ,q >0$ (depending only on $p$) 
		such that for all $\ve>0 $ and
		$\Psi \in   C^\infty(\To^3) \cap \mathscr H_B  $ 
		\begin{equation}
			\label{eq:h}
			(1   -  \ve ) 
			\< \Psi, h_0 \Psi \> 
			-   \frac{  c  N  }{\ve^a}  \| w\|^q_{L^p}  
			\|	 \Psi	\|^2
			\leq 
			\< \Psi, h \Psi  \>
			\leq (1 + \ve ) 
			\< \Psi, h_0 \Psi \> 
			+  \frac{ c N   }{\ve^a}    \| w\|^q_{L^p}
			\|	 \Psi	\|^2 \ . 
		\end{equation}
		In particular, $h$ has a unique self-adjoint extension on $\mathscr H_B$  with form domain 
		$  Q(h_0 )$ and \eqref{eq:h} 
		holds for all $\Psi \in Q(h_0)$. 
	\end{lemma}
	\begin{remark}
		From the proof, one may compute $q  = \frac{2 p }{ 2p-3}$. 
	\end{remark}

	\begin{proof}
		Consider first $\vp \in H^1 (\To^3_x \times \To^3_y)$.
		For  $     \tfrac32 <  p \leq \infty $
		let $p ' = \frac{p}{p-1}$ be the dual exponent, with
		$ 1  \leq   p ' < 3  $.
		Thanks to H\"older's inequality, $L^p$ interpolation and  the 
		Sobolev embedding  $H^1 (\To^3) \hookrightarrow L^6 (\To^3)$,   we  have  
		for $\vp_x(y) = \vp(x, y)$
		\begin{align*}
			\iint   w (x -y )  | \vp(x, y)|^2     dy dx 
			& 	     \leq \int   \| w\|_{L^p}	\|	 \vp_x	\|_{L^{2 p '}}^2  dx  \\
			& 	     \leq  C	 \| w\|_{L^p}	  \int 
			\|  \vp_x \|^{2 \theta_1}_{L^2}	 \|   \vp_x \|_{L^6}^{2 \theta_2}  dx 
			\leq  C	 \| w\|_{L^p}	  \int    	 \|  \vp_x \|^{2 \theta_1}_{L^2}	 
			\|  \nabla  \vp_x \|_{L^2}^{2\theta_2}   dx \ .
		\end{align*}
		Here $ \theta_1 =  \frac{2}{2p'} \frac{ 6 - 2p' }{6-2 } = \frac{2p-3}{2 p}$
		and
		$\theta_2  = \frac{6 }{2 p ' }  \frac{ 2p' - 2 }{6-2} = \frac{3}{2p} $
		are found by interpolation. 
		Thus,  we may apply Young's inequality to find 
		there is  $C = C(p)>0$
		such that for all $\ve>0$
		\begin{equation}
			\label{eq:energy}
			\int   \int   w (x -y )  | \vp(x, y)|^2     dy dx 
			\leq 
			\ve 	   
			\|    \nabla \vp	\|_{L^2}^2 
			+ 
			C 
			\ve^{- \theta_2 / \theta_1}  \| w\|_{L^p}^{ 1 / \theta_1 }
			\| \vp \|_{L^2}^2 
		\end{equation} 
		Next, we turn to the proof of the lemma. Thanks to permutational symmetry, 	 it suffices to estimate the quadratic form
		\begin{equation}
			\<  \Psi,   	 \frac{1}{N}  \sum_{ i<j }w (x_i - x_j) \Psi	\>
			= 
			\frac{N-1}{2}
			\int d x_3 \cdots d x_N  
			\int d x_1 d x_2 w (x_1 - x_2) |	 \psi  (x_1, x_2)	|^2 dx_1 dx_2
		\end{equation}
		with $\psi (x_1, x_2) \equiv \Psi(x_1, x_2,  x_3, \cdots, x_N)$. 
		We may now use \eqref{eq:energy} for $\vp = \psi$ and then integrate over the remaining variables. 
		Thus, we conclude the potential energy is relative form bounded with respect to $h_0$, 
		with constant $\ve <1$. Thus, thanks to the KLMN theorem the operator $h$  
		has a unique self-adjoint extension, 
		and  the desired inequality follows by a standard approximation argument. 
	\end{proof}

	\section{Proof of Corollary \ref{coro}}
	\label{appendix:projections}

	Before we turn to the proof of Corollary \ref{coro}   we argue that $h^\eff $ has  a non-degenerate ground state.
	Indeed, consider on  the whole Hilbert space $L^2 ( \To^{3  N })$ the 
	Schr\"odinger operator 
	$$ \mathcal H  =\textstyle 
	\sum_{i =1 }^N ( - \Delta_{x_i}) + N^{-1}  \sum_{ i< j} W^{\eff} (x_i - x_j ) \ .
	$$
	One may    adapt the proof of  \cite[Theorem 11.8]{Lieb:book} to the torus case
	$\To^{3N}$
	and  show that 
	$\mathcal H $ has  a ground state  $\phi^\eff \in L^2 (\To^{3N})$ which is strictly positive and unique up to a constant phase (for this,  we adapt the  convexity inequality for gradients in \cite[Theorem 7.8]{Lieb:book} to the torus). 
	By standard arguments, the ground state is also permutationally symmetric, i.e. $\phi^{\eff} \in \mathscr H_B$.  
	As a consequence, we conclude  that for each $ N \geq 1 $ we have the spectral gap 
	\begin{equation}\label{eq:spectral:gap}
		\mu_2 ( h^\eff) - \mu_1  (h^\eff) > 0  \ . 
	\end{equation}

	\begin{proof}[Proof of Corollary \ref{coro}] 
		Fix $N\geq1$ throughout the proof. We prove the statement in the particle-hole picture. First, recall $\Omega = \mathscr R^* \Omega_F$ and consider the trial state $\Psi = (1-\lambda \mathbb R\mathbb V^+) \Psi_0$, $\Psi_0 = \phi^{\rm eff} \otimes \Omega$. Note that this state is analogous to the trial state introduced in \eqref{psi} (for $n=1$) with $\Phi \equiv \phi^{\eff}$. Second, following the exact same steps leading to \eqref{eq4.8}, and using \eqref{supremum:difference} and \eqref{eq1} to relate $h_{k_F}^\eff$ and $h^\eff$  one   verifies that
		\begin{align}
			\label{trial:state:bound:appendix}
			\<
			\Psi, \H \Psi
			\> 
			\ 		\leq \  \mu_1 (	\widetilde h 	)  \ + 
			o (1) \qquad \kf \rightarrow \infty 
		\end{align}
		where  $ \widetilde h  \equiv h^{\rm eff} - \tfrac12 V\ast V(0)$ for notational convenience. Third, we introduce the orthogonal projections $P = |\mathscr R^*\Phi \rangle \langle\mathscr R^* \Phi |$ and $Q = 1 -P$, and let $ \triangle \equiv \mu_2 ( h^\eff) - \mu_1  (h^\eff)$.
		We estimate
		\begin{align*}
			& \langle \Psi  , (\mathbb H-\mu_1( \widetilde h ) )  \Psi \rangle
			\  \ge  \ 
			\| P \Psi \|^2 \big(\mu_1(\mathbb H)  \ - \   \mu_1( \widetilde h  ) \big) 
			\ + \ 
			\| Q \Psi\|^2 \big(\mu_2(\mathbb H)  \ - \  \mu_1(  \widetilde h  )\big)
		\end{align*}
		where we used $\mathbb H P = \mu_1(\mathbb H)P$, $Q\mathbb HP=0=QP$, and $Q\mathbb HQ \ge \mu_2(\mathbb H)Q$. 
		Applying Propositions \ref{prop2} and \ref{prop1} for $n=1$ and $n=2$, respectively, together with \eqref{eq1}, we find
		\begin{align*}
			\langle \Psi , (\mathbb H-\mu_1 (    \widetilde h   ) )  \Psi \rangle 
			\ \ge  \ 
			\|Q\Psi \|^2 \big( \mu_2( \widetilde h   )  \ - \  \mu_1(  \widetilde h  ) \big)   \ - \   
			o (1) \qquad \kf \rightarrow \infty  \ . 
		\end{align*}
		Combining this with \eqref{trial:state:bound:appendix}  we obtain 
		\begin{align*}
			\triangle \|Q\Psi \|^2  &
			\,  \le \,
			\langle \Psi , (\mathbb H-\mu_1(
			\widetilde h   ) ) \Psi \rangle 
			+ o(1)   \, =  \,  o(1)  \, \
			\qquad 
			\,  \kf \rightarrow \infty  \ . 
		\end{align*}
		Using Lemma \ref{lemma:partial1} and the definition of $\Psi$, we further estimate
		\begin{align*}
			\| Q (\phi^{\rm eff}  \otimes \Omega) \| 
			\ \le \ 
			\| Q\Psi \| + \| \lambda \mathbb R \mathbb V^+ (\phi^{\rm eff} \otimes \Omega) \| 
			\ \le  \  \| Q\Psi \|  \  + \ 
			o (1)  \ , \qquad \kf \rightarrow \infty 
		\end{align*}
		and thus
		\begin{align*}\label{eq:Q:bound}
			\| Q(\phi^{\rm eff}  \otimes \Omega) \|^2 
			=
			o (1) \ , \  \qquad \kf \rightarrow \infty  \ . 
		\end{align*}
		This implies $| \langle \mathscr R^* \Phi , \phi^{\rm eff} \otimes \Omega \rangle | = | \langle  \Phi , \phi^{\rm eff} \otimes \Omega_{\rm F} \rangle | \to 1$. 
		This gives 
		\begin{align*}
			\|	  \mathscr R^* \Phi -  \alpha  \phi^\eff \otimes \Omega		\|^2 
			=   2 (	 1     - | \langle  \Phi , \phi^{\rm eff} \otimes \Omega_{\rm F}  
			\rangle| ) 
			\rightarrow 0 \qquad \kf\rightarrow \infty 
		\end{align*}
		for $ \alpha = | \langle  \Phi , \phi^{\rm eff} \otimes \Omega_{\rm F} \rangle |
		\langle  \Phi , \phi^{\rm eff} \otimes \Omega_{\rm F} \rangle^{-1}$. Choosing the relative phase so that $ \langle  \Phi , \phi^{\rm eff} \otimes \Omega_{\rm F} \rangle \ge 0$, we obtain $\alpha =1 $, which completes the proof of the corollary.
	\end{proof}

	\medskip
	
	\noindent \textbf{Acknowledgements.} The authors would like to thank Robert Seiringer for valuable comments. 
	E.C. gratefully acknowledges support from NSF under grant DMS-2052789 through Nata\v sa Pavlovi\'c. 
	N.P. gratefully acknowledges support from the NSF under grants No. DMS-1840314 and DMS-2052789

	\vspace{3mm}
	
	\noindent \textbf{Data availability.} This manuscript has no associated data.
	\vspace{3mm}
	
	\noindent 
	\textbf{Conflict of interest.} 
	The authors state  that there is no conflict of interest.

	\end{document}